\documentclass[authorcolumns,numberwithinsect,a4paper]{no-lipics}
\newcommand{\oh}[1]{o(#1)}
\newcommand{\Oh}[1]{\mathcal{O}(#1)}
\newcommand{\pspace}{\ensuremath{\mathsf{PSPACE}}}

\newcommand{\N}{\ensuremath{\mathbb{N}}}

\newcommand{\trias}[1]{T(#1)}
\newcommand{\triasCap}[1]{T(#1)}
\newcommand{\edgeCap}[2][H]{E_{#1}(#2)}
\newcommand{\etgraph}[1]{\mathsf{ET}(#1)}
\newcommand{\vegraph}[1]{\mathcal{I}(#1)}
\newcommand{\gamehg}[1]{\mathcal{H}(#1)}
\newcommand{\concomp}[3][]{\kappa_{#1}(#2,#3)}

\newcommand{\generalTime}{n+m^3+m^2t}
\newcommand{\generalTimeNoTrig}{n+m^{3.5}}
\newcommand{\cliqueTimeBase}{n+m^{1.5}}
\newcommand{\cliqueTimeAddition}{\min\{n^{\omega+1},m^2\}}
\newcommand{\cliqueTime}{n+\min\{n^{\omega+1},m^2\}}
\newcommand{\cactusTime}{n+m^{1.5}}

\title{Faster Algorithms for Deciding the Unbiased Maker-Breaker Triangle Game on General Graphs}
\titlerunning{Faster Algorithms for the Unbiased Maker-Breaker Triangle Game}
\author{Julian Christoph Brinkmann}{Goethe University Frankfurt, Germany}{J.Brinkmann@em.uni-frankfurt.de}{https://orcid.org/0009-0000-0332-4543}{}
\author{Anand Srivastav}{Kiel University, Germany\\ Goethe University Frankfurt, Germany}{srivastav@math.uni-kiel.de}{}{}
\authorrunning{J. C. Brinkmann, A. Srivastav}

\begin{document}

\maketitle
\pagenumbering{roman}
\begin{abstract}
	In this paper, we present new polynomial-time algorithms for determining the winner of the unbiased triangle game played on the edge set of general graphs.
	To that end, we propose to view the game through the edge-triangle incidence graph instead of the standard hypergraph model.
	We identify a necessary and sufficient winning condition for Maker in terms of the edge-triangle incidence graph and show that winning strategies achieving this condition as fast as possible play monotonically in the sense that they only consider monotonically decreasing connected subgraphs of the edge-triangle incidence graph.
	
	We give three different algorithms for different classes of graphs.
	For general graphs \(G\), the outcome of the unbiased triangle game can be decided in time \(\Oh{\generalTimeNoTrig}=\Oh{n^7}\).
	This significantly improves on the \(\Oh{n^{16}}\) algorithm implied by the work of Galliot, Gravier and Sivignon (arXiv 2022).
	For graphs \(G\) which contain \(K_4\), the complete graph on four vertices, as a subgraph and whose edge-triangle incidence graph is connected, the winner can be decided in time~\(\Oh{\cliqueTime}=\Oh{n^{\omega+1}}\), where \(\omega<2.372\) is the exponent of matrix multiplication~{(Alman et.~al., SODA 2025)}.
	For graphs \(G\) whose edge-triangle incidence graph is a cactus graph, i.e.~all its cycles are edge-disjoint, the winner can be decided in time~\(\Oh{\cactusTime}=\Oh{n^3}\).
	Such \(G\) are \(K_4\)-free.
	The algorithms for the special cases are based on novel structural characterizations of Maker's win for each graph class.
	
	We give a linear time reduction from triangle detection to deciding the unbiased triangle game.
	The best currently known algorithm for triangle detection runs in time \(\Oh{n^\omega}\).
	
	Additionally, we show that, asymptotically, any graph with \(\left(\frac{1}{4}+\varepsilon\right)n^2\) edges is Maker's win.
	This is best possible.
	Finally, we give a tight lower bound of \(2n-1\) on the number of edges in any graph that is minimally winning for Maker with respect to graph inclusion.
\end{abstract}

\thispagestyle{plain}
\setcounter{tocdepth}{2}
\tableofcontents
\clearpage
\pagenumbering{arabic}

\section{Introduction}
In a positional game, two players alternatingly claim vertices of a hypergraph.
There are different regimes for determining the winner of these games, such as Maker-Maker games, Maker-Breaker games, Avoider-Enforcer games, see~\cite{hefetz2014positional} for an overview of the field.
We consider Maker-Breaker games which were introduced by Erdős and Selfridge~\cite{DBLP:journals/jct/ErdosS73}.
In this setting, the first player, Maker, tries to claim all vertices of an edge of the hypergraph and the second player, Breaker, tries to prevent this.
Maker-Breaker games are often played on edge sets of graphs, i.e.~the vertices of the hypergraph are the edges of some graph and the hyperedges are collections of edges inducing some predetermined structure, such as spanning trees, perfect matchings, Hamiltonian cycles or copies of a fixed target graph \(H\).
In this paper, we study the Maker-Breaker \(K_3\)-game, the unbiased triangle game, played on general graphs.

The study of Maker-Breaker sugbraph was initiated by the pioneering work of Chvatal and Erdős~\cite{chvatal1978biased}. 
There are two main variants of the game, the biased version played on the edges of complete graphs and the unbiased game played on the edges of general graphs or hypergraphs.
In the first variant, Maker is usually allowed to choose one edge of the complete graph, while Breaker may choose up to \(q\) edges.
The goal is to find the threshold bias of the game, i.e.~the minimum \(q\) so that Breaker wins.
In this setting, \(q\) is analyzed as a function of the number of vertices of the board graph.
For triangles, the biased game has been extensively studied, see e.g.~\cite{chvatal1978biased,DBLP:journals/combinatorics/BaloghS11,DBLP:journals/ejc/GlazikS22}.
In the unbiased game, \(q=1\).
Here, the challenge in understanding the game arises from the fact that the board is a general graph leading to complexity-theoretic barriers.

The complexity classification of Maker-Breaker games played on hypergraphs has recently been resolved.
A hypergraph has rank \(k\) if every hyperedge has size at most \(k\) and there exists a hyperedge of size exactly \(k\).
In a fundamental work, Schaefer~\cite{DBLP:journals/jcss/Schaefer78} proved that deciding the winner of a Maker-Breaker game is \pspace-complete even restricted to hypergraphs of rank at most~11.
This result was improved to hypergraphs of rank at most~6~\cite{DBLP:journals/combinatorica/RahmanW23} and recently to hypergraphs of rank at most~5~\cite{koepke2025solvingmakerbreakergames5uniform}.
Finally, Galliot \cite{DBLP:journals/corr/abs-2509-13819} further improved the rank bound showing that hypergraphs of rank at most~4 are sufficient for \pspace-completeness.
On the other hand, the outcome of Maker-Breaker games on hypergraphs of rank at most~3 can be determined in polynomial time~\cite{galliot2025makerbreakersolvedpolynomialtime}.
This improves upon the work of Kutz~\cite{kutz2004thesis} who showed that the outcome of Maker-Breaker games on linear hypergraphs of rank at most~3 can be determined in polynomial time.

The aforementioned publications study the complexity of the general Maker-Breaker game played on hypergraphs.
The Maker-Breaker \(H\)-game was investigated by Duchéne et al.~\cite{DBLP:journals/dam/DucheneGINOPS25}.
They show that the Maker-Breaker \(T\)-game is \pspace-complete even when played on graphs~\(G\) with diameter at most~6, where \(T\) is a particular tree with~91 vertices.
Additionally, the Maker-Breaker \(H\)-game is shown to be \pspace-complete for a particular graph \(H\) with~51 vertices and~57 edges.
Their reductions are based on the hardness of the Maker-Breaker game played on hypergraphs of rank~6.
With the recent improvement of~\cite{DBLP:journals/corr/abs-2509-13819}, the size of \(T\) and \(H\) can be further reduced to~45 vertices and~35 vertices with~39 edges, respectively.

The algorithm of~\cite{galliot2025makerbreakersolvedpolynomialtime} runs in time \(\Oh{|V(\mathcal{H})|^5|E(\mathcal{H})|^2+|V(\mathcal{H})|^6\Delta(\mathcal{H})}\) where \(\Delta(\mathcal{H})\) denotes the maximal degree of a vertex in the hypergraph \(\mathcal{H}\).
This implies that the outcome of the \(H\)-game can be determined in polynomial time for graphs \(H\) with at most three edges.
However, the asymptotic running time is \(\Oh{n^{16}}\) as the vertices of \(\mathcal{H}\) are the edges of the board graph \(G\) and the edges of \(\mathcal{H}\) correspond to the edge-induced subgraphs of \(G\) isomorphic to \(H\).
Duchéne et al.~\cite{DBLP:journals/dam/DucheneGINOPS25} determine the maximal winning graphs for Breaker in the \(P_4\)-game, where \(P_4\) is a path consisting of four vertices.
Membership in this family of maximal winning graphs can be checked in linear time given the board graph \(G\) resulting in a linear time algorithm for the \(P_4\)-game.
They also show that the \(K_{1,\ell}\)-game is fixed-parameter tractable when parameterized by the length of the game, where \(K_{1,\ell}\) is a star graph with \(\ell\) edges.
The case \(\ell=3\) is trivially fixed-parameter tractable by the algorithm of~\cite{galliot2025makerbreakersolvedpolynomialtime}.

\subsection{Our Contribution}
The only interesting graph \(H\) with at most three edges not considered in~\cite{DBLP:journals/dam/DucheneGINOPS25} is the cycle of length three.
Maker's win in the \(H\)-game can easily be characterized when \(H\) contains at most two edges or is not connected.
We give lower and upper bounds on the time complexity of the unbiased triangle game played on general graphs.

We present three different algorithms for different classes of graphs.
In the following, \(n\) denotes the number of vertices of the input graph and \(m\) denotes the number of edges of the input graph.
Our algorithm for general graphs has running time \(\Oh{\generalTime}=\Oh{n^7}\), where \(t\) denotes the number of subgraphs of \(G\) isomorphic to \(K_3\), the complete graph on~3 vertices.
The time bound for the algorithm of~\cite{galliot2025makerbreakersolvedpolynomialtime} corresponds to \(\Oh{n+m^5t^2+nm^6}=\Oh{n^{16}}\).
Thus, we significantly improve upon the algorithm of~\cite{galliot2025makerbreakersolvedpolynomialtime}.
For further improvements of the running time, our rationale is to distinguish whether the input graph \(G\) contains \(K_4\) or is \(K_4\)-free.
If the edge-triangle incidence graph of \(G\) (see \cref{defn:etg}) is connected and \(G\) contains \(K_4\) as a subgraph, the outcome of the game can be determined in time \(\Oh{{\cliqueTime}}=\Oh{n^{\omega+1}}\), where~\(\omega<2.372\) is the exponent of matrix multiplication~\cite{DBLP:conf/soda/AlmanDWXXZ25}..
This result allows future work to focus on \(K_4\)-free graphs using the rich theory of extremal graphs.
If the edge-triangle incidence graph of \(G\) is a cactus graph, the winner of the game can be determined in time~\({\Oh{\cactusTime}=\Oh{n^3}}\).
Such \(G\) are \(K_4\)-free.
Our results as well as previous algorithms are summarized in \cref{fig:running-time-table}.

\begin{figure}
	\centering
	\begin{tabular}{c|l|c}
		Graph Class & \multicolumn{1}{c|}{Running time} & Source\\\hline
		general graphs & \(\Oh{\text{poly}(n)}\) & \cite{kutz2004thesis}\\
		general graphs & \(\Oh{n+m^5t^2+nm^6}\) & \cite{galliot2025makerbreakersolvedpolynomialtime}\\
		general graphs & \(\Oh{\generalTime}\) & this work\\
		\(K_4\subseteq G\), \(\etgraph{G}\) connected & \(\Oh{\cliqueTime}\) & this work\\
		\(\etgraph{G}\) cactus graph & \(\Oh{\cactusTime}\) & this work\\
	\end{tabular}
	\caption{Overview of algorithms for the unbiased triangle game on different graph classes.
	The author of~\cite{kutz2004thesis} does not provide a concrete running time, arguing only that is polynomial.}
	\label{fig:running-time-table}
\end{figure}

To obtain a lower bound for the time complexity, we present an algorithm that, given a graph \(G\), constructs a graph \(H\) with \(|V(H)|=n+2m\) and \(|E(H)|=6m\) in time \(\Oh{n+m}\) such that \(G\) contains a triangle if and only if Maker wins the unbiased triangle game on \(H\), see \cref{thm:trig-detection-lb}.
This shows that the unbiased triangle game is at least as hard as triangle detection.
The asymptotically fastest known algorithm for triangle detection uses matrix multiplication~\cite{DBLP:journals/siamcomp/ItaiR78} and runs in time \(\Oh{n^\omega}\).
This gives evidence that he unbiased triangle game cannot be decided in time \(\oh{n^{\omega}}\).

A key insight underlying our algorithms is that it is useful to model the unbiased triangle game via the edge-triangle incidence graph instead of the hypergraph representation.
This enables us to compute connected and even 2-connected components in linear time using classical algorithms.
The high running time of the algorithm of Galliot, Gravier and Sivignon~\cite{galliot2025makerbreakersolvedpolynomialtime} is partly caused by relying on chain detection in hypergraphs as a subroutine.
This is completely avoided through the framework of the edge-triangle incidence graph.
To further improve the running time, we exploit the particular structure of the unbiased triangle game.
The hypergraphs modeling the unbiased triangle game are linear allowing us to use results by Kutz~\cite{kutz2004thesis}.
Additionally, we utilize the special role of subgraphs isomorphic to \(K_4\).
Such subgraphs are very threatening to Breaker in a sense which is made precise in~\cref{thm:k4-elimination}.

The paper is structured as follows:
In \cref{sec:prelims}, we introduce the terminology used in this work.
Then, we analyze winning strategies for Maker and Breaker through the lens of the edge-triangle incidence graph in \cref{sec:strategy-analysis}.
This analysis lays the foundations for the following sections, but also contains results of independent interest such as a new kind of monotonicity property, see \cref{thm:monotonicity}.
In \cref{sec:complexity}, we present and analyze the algorithms for the unbiased triangle game on the three mentioned classes of graphs.
For the two special cases, this requires specialized structural characterizations of Maker's win, see \cref{thm:k4-elimination,thm:cactus-game-chara}.
In \cref{sec:min-maker-win}, we determine the asymptotic number of edges required to ensure Maker's win and lower bound the number of edges in graphs that are minimally winning for Maker with respect to graph inclusion.
Finally, we discuss open problems in \cref{sec:opem-problems}.

\section{Preliminaries}\label{sec:prelims}
Let \(G\) be a graph.
We denote the vertex set of \(G\) by \(V(G)\) and the edge set of \(G\) by \(E(G)\).
For hypergraphs, we use the same notation.
We denote the \emph{neighborhood} of \(U\subseteq V(G)\) in \(G\) by \(N_G(U)\).
Formally, \({N_G(U)=\{w\in V(G): vw\in E(G), v\in U\}\setminus U}\).
Similarly, the \emph{closed neighborhood} of \(U\) is \({N_G[U]=N_G(U)\cup U}\).
We will omit the subscript if \(G\) is clear from the context.
For a vertex \(v\in V(G)\), we write \(N(v)\) and \(N[v]\) instead of \(N(\{v\})\) and \(N[\{v\}]\).

For \(U\subseteq V(G)\), the \emph{subgraph induced by \(U\)} is denoted by \(G[U]=(U,E_U)\) where \(E_U\) is the set of edges using only vertices from \(U\), i.e.~\(E_U=\{e\in E(G)\mid e\subseteq U\}\).
We use \(G-U\) to denote the induced subgraph \(G[V(G)\setminus U]\).
For a vertex \(v\in V(G)\), we write \(G-v\) instead of \(G-\{v\}\).
For \({F\subseteq E(G)}\), \emph{the subgraph induced by \(F\)} is denoted by \(G[F]=(V_F,F)\) where \(V_F\) is the set of endpoints of edges in \(F\).
We use \(G\setminus F\) to denote the graph \((V(G),E(G)\setminus F)\).
The graph \(G\setminus F\) may be different from \(G[E(G)\setminus F]\).
For an edge \(e\in E(G)\), we write \(G\setminus e\) instead of \(G\setminus\{e\}\).

A graph is a \emph{cactus graph} if all its cycles are edge-disjoint.
For \(n\in\N\), the graph \(K_n\) is the complete graph on \(n\) vertices.
Two graphs \(G\) and \(H\) are \emph{isomorphic} if and only if there exists a bijection \({\varphi:V(G)\to V(H)}\) such that for all \(v,w\in V(G)\) the property \(vw\in E(G)\) holds if and only if \(\varphi(v)\varphi(w)\in E(H)\) holds.
We use \(G\cong H\) to denote that \(G\) and \(H\) are isomorphic.
Edge sets that induce an isomorphic copy of \(K_3\) in \(G\) are called \emph{triangles}.
We denote the \emph{set of triangles} in \(G\) by \(\trias{G}\).
Formally, \(\trias{G}={\{F\subseteq E(G):G[F]\cong K_3\}}\).
For \(U\subseteq V(G)\), \(\trias{U}=\trias{G[U]}\).

Given a hypergraph \(H\) and a subset \(U\subsetneq V(H)\), the hypergraph \(H^{+U}\) has vertex set \(V(H)\setminus U\) and edge set \({\{e\setminus U:e\in E(H)\}\setminus\{\emptyset\}}\).
For \(v\in V(H)\), we write \(H^{+v}\) instead of \(H^{+\{v\}}\).
For \(\emptyset\neq e\subseteq V(H)\), the hypergraph \(H+e\) has vertex set \(V(H)\) and edge set \(E(H)\cup\{e\}\).

We denote by \(\omega\) the \emph{exponent of matrix multiplication}, which has been intensively studied for decades~\cite{strassen1969gaussian,DBLP:conf/focs/Pan78,DBLP:journals/ipl/BiniCRL79,DBLP:journals/siamcomp/Schonhage81,DBLP:journals/siamcomp/Romani82,DBLP:conf/focs/CoppersmithW81,DBLP:conf/stoc/Williams12,DBLP:conf/soda/AlmanDWXXZ25}.
The best currently known upper bound for \(\omega\) is \(\omega<2.372\)~\cite{DBLP:conf/soda/AlmanDWXXZ25}.

The \emph{unbiased triangle game} played on the graph \(G\) is the following two player game:
The players Maker and Breaker alternatingly claim edges of \(G\) starting with Maker.
In each round, Maker and Breaker claim exactly one edge that has not already been claimed by any player as long as there are still edges available.
After any round, Maker wins if a subset of the edges he claimed induces a \(K_3\) in \(G\).
If no more edges are left and no subset of Maker's edges induces a \(K_3\), Breaker wins.
We will sometimes call the graph \(G\) the \emph{board graph} to emphasize that this is the graph the game is played on.
If Maker wins on \(G\), \(G\) is called \emph{Maker's win}.
Otherwise, \(G\) is called \emph{Breaker's win}.

Our definition of the unbiased triangle game is slightly different from the standard definition in the literature:
Usually, Maker and Breaker claim edges until there are none left.
Then, the winner is decided based on the edges claimed by Maker.
We allow the game to end early if Maker has already claimed a triangle.
This difference in definition makes it possible to concisely state monotonicity properties of winning strategies for Maker in \cref{sec:strategy-analysis}.

A \emph{configuration} of the game is a pair \((E_m,E_b)\) where \(E_m,E_b\subseteq E(G)\) and \(E_b\cap E_m=\emptyset\).
It represents the current state of the game, Maker has claimed all the elements of \(E_m\) and Breaker has claimed all the elements of \(E_b\).
If Maker claims the edges \(e_1,\dots,e_\ell\) in this order and Breaker replies with \(a_1,\dots,a_\ell\) in this order, then this particular play can be viewed as the sequence of configurations \((\emptyset,\emptyset),(E_m^1,E_b^1),\dots,(E_m^\ell,E_b^\ell)\) where \(E_m^i=\{e_j:1\leq j\leq i\}\) and \(E_b^i=\{a_j:1\leq j\leq i\}\).

The unbiased triangle game with board graph \(G\) can also be modeled by the hypergraph \(\gamehg{G}\) with \({V(\gamehg{G})=E(G)}\) and \({E(\gamehg{G})=\{U\subseteq E(G)\mid G[U]\cong K_3\}}\).
The unbiased triangle game is equivalent to the unbiased Maker-Breaker game played on \(\gamehg{G}\).
Two different triangles can only intersect in at most two vertices, or, equivalently, a single edge.
Therefore, the hypergraph \(\gamehg{G}\) is linear.
To reduce the amount of used terminology, we will usually say that Maker or Breaker wins on \(\gamehg{G}\) instead of referring to the unbiased triangle game on \(G\).

\begin{figure}[t]
	\centering
	\hspace{15mm}
	$\vcenter{\hbox{\includegraphics[scale=0.7]{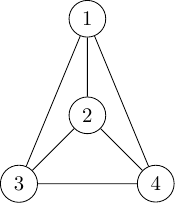}}}$
	\hspace{3cm}
	$\vcenter{\hbox{\includegraphics[scale=0.7]{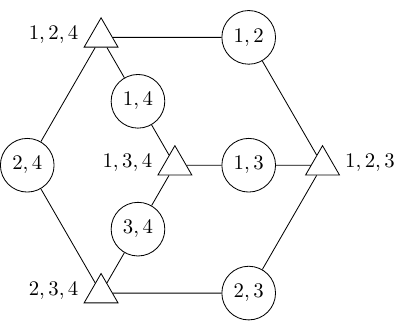}}}$
	\hspace{15mm}
	\caption{The graph \(K_4\) (left) and its edge-triangle incidence graph (right).
		In the edge-triangle incidence graph, circular vertices correspond to edges of \(K_4\) and triangular vertices correspond to triangles in \(K_4\).
		The vertex label indicates the vertices involved in the edge or triangle.}
	\label{fig:etgraph-example}
\end{figure}

For algorithmic purposes, we model the game by the following incidence structure:
\begin{definition}\label{defn:etg}
	Let \(G\) be a graph.
	The edge-triangle incidence graph of \(G\), denoted by \(\etgraph{G}\), has vertex set \(E(G)\cup\trias{G}\) and edge set \(\{\{e,t\}:e\in t\in\trias{G}\}\).
\end{definition}
All triangle-vertices have degree exactly~3 in \(\etgraph{G}\).
The edge-triangle incidence graph of the complete graph on four vertices is given as an example in \cref{fig:etgraph-example}.
Any Maker-Breaker game played on a hypergraph \(H\) can also be modeled by the vertex-hyperedge incidence graph \(\vegraph{H}\) with \({V(\vegraph{H})=V(H)\cup E(H)}\) and \({E(\vegraph{H})=\{\{v,e\}\mid v\in e\in E(H)\}}\) generalizing \cref{defn:etg}.
Maker and Breaker still alternatingly claim elements of \(V(H)\).
The goal of Maker now is to claim the complete neighbor set \(N_{\vegraph{H}}(e)\) of an edge \(e\in E(H)\) while Breaker tries to prevent this.
Similarly, any incidence structure can be modeled by a hypergraph.
In this sense, the two views of the game are the same.
However, if the hypergraph is linear, it is useful to work with graphs instead of hypergraphs.
This is because shortest paths and cycles in \(\vegraph{H}\) correspond to paths and cycles in \(H\) while being easier to compute.

Given a path or cycle \(R\) in \(\etgraph{G}\), we will use \(\triasCap{R}\) as shorthand notation for \(V(R)\cap\trias{G}\).
To distinguish between edges of \(G\) and vertices in \(\etgraph{G}\), we will call elements of \(E(G)\) edge-vertices when we view them as vertices of \(\etgraph{G}\).
Similarly, we will call elements of \(\trias{G}\) triangle-vertices when we view them as vertices of \(\etgraph{G}\).
For a hypergraph \(H\) and a path \(P\) in \(\vegraph{H}\), we use \(\edgeCap{P}\) to denote \(V(P)\cap E(H)\).

\section{Winning Strategies and the Edge-Triangle Incidence Graph}\label{sec:strategy-analysis}
In this section, we analyze properties of winning strategies for Maker and Breaker from the perspective of the edge-triangle incidence graphs.
We characterize Breaker's win in \cref{lem:maker-winning-structure}, which also gives a sufficient and necessary condition for Maker's win.
Then, we analyze properties of winning strategies for Maker that play to achieve this winning condition as fast as possible.
Notably, such strategies play monotonically, i.e.~they only claim edges that lie in monotonically decreasing connected subgraphs of \(\etgraph{G}\), see \cref{thm:monotonicity}.
We also show that such strategies can be assumed to only claim edges in round \(i+1\) that lie on a cycle in \(\etgraph{G}-N[E_b^i]\), where \(E_b^i\) denotes that Breaker has claimed up to round \(i\) with the convention that \(E_b^0=\emptyset\).

\begin{lemma}\label{lem:maker-winning-structure}
	Let \(B\) be a strategy for Breaker on \(\gamehg{G}\).
	Then \(B\) is a winning strategy if and only if for every Maker strategy \(M\) the following condition is satisfied:
	Denote the set of edges Breaker has claimed after round \(i\) by \(E_b^i\) while playing against \(M\).
	Then after each round \(i\), Maker has claimed at most one edge-vertex in each connected component of \(\etgraph{G}-N[E_b^i]\).
	
	In particular, if Maker has a strategy on \(\gamehg{G}\) that always manages to claim two edge-vertices in the same connected component of \(\etgraph{G}-N[E_b^i]\) against any Breaker strategy for some \(i\), then Maker wins on \(\gamehg{G}\).
\end{lemma}
\begin{proof}
	The second statement follows immediately from the first as no Breaker strategy can be a winning strategy if a Maker strategy with this property exists.
	We now show the first statement.
	The condition is clearly sufficient for \(B\) to be a winning strategy for Breaker:
	Assume for contradiction that some strategy \(M\) manages to claim a triangle \(t\).
	Then \(M\) has claimed all three edges of \(t\).
	In particular, \(M\) has claimed 3 edge-vertices in the same connected component contradicting the assumption.
	
	It remains to show that the condition is necessary.
	Assume that there exists a strategy \(M\) such that in some round \(i\), \(M\) manages to claim two edge-vertices in the same connected component of \(\etgraph{G}-N_b^i\).
	Let \(j\) be the first round where this happens.
	The round \(j\) cannot be the last round as edge-vertices can only be connected by paths containing triangles, so there must be triangles with unclaimed edges left.
	We will construct a strategy \(M^*\) that wins against \(B\) concluding the proof.
	
	For the first \(j\) rounds, \(M^*\) plays according to \(M\).
	Let \(P=e_1t_2e_2\dots t_ke_k\) be a shortest path between two edge-vertices claimed by Maker in \(\etgraph{G}-N[E_b^i]\) where \(e_1,\dots,e_k\in E(G)\) and \(t_2,\dots,t_k\in\trias{G}\).
	The edges \(e_2,\dots,e_{k-1}\) are not claimed by any player.
	As \(P\) is shortest, there are no edges between vertices of \(V(P)\) except those used by the path.
	Additionally, for \(2\leq a,b\leq k\), \(t_a\) and \(t_b\) have no common neighbors if \(|a-b|>1\) as then \(P\) also would not be shortest.
	If \(|a-b|=1\), \(t_a\) and \(t_b\) have exactly one common neighbor as otherwise they would be identical.
	This is due to the fact that the vertices of a triangle are already determined by two edges.
	
	Next, \(M'\) claims \(e_2,e_3,\dots,e_{k-1}\) in this order.
	In every turn, Breaker must respond to \(e_i\) by claiming the unique edge in \(t_i\setminus V(P)\).
	The above argument shows that all of these edges are pairwise distinct.
	If Breaker does not respond in this way, Maker wins by claiming the edge himself in the next round.
	Finally, Maker wins by claiming the remaining edge in \(N(t_k)\setminus V(P)\).
	The situation is visualized in \cref{fig:path-strategy}.
\end{proof}

\begin{figure}[t]
	\centering
	\includegraphics[scale=0.7]{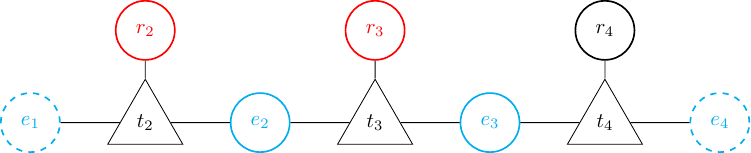}
	\caption{A path of length 7 on which Maker has claimed two edge-vertices (dashed) and the resulting edge-vertices that Maker and Breaker claim in the proof of \cref{lem:maker-winning-structure}.
	Edge-vertices claimed by Breaker are colored red.
	Cyan edge-vertices are claimed by Maker.
	Black edge-vertices are unclaimed.}
	\label{fig:path-strategy}
\end{figure}

The Maker strategy used in the proof of \cref{lem:maker-winning-structure} is suboptimal with respect to the number of rounds needed to win.
By claiming the edge-vertex in the middle of the current path, all edges of a triangle can be claimed in a logarithmic number of rounds with respect to the length of the initial path.
This idea is used in~\cite{galliot2025makerbreakersolvedpolynomialtime} to show that either Maker can win the unbiased Maker-Breaker game on a hypergraph \(H\) of rank at most~3 within \(\Oh{\log |V(H)|}\) rounds or \(H\) is Breaker's win.

\begin{definition}
	Let \(G=(V,E)\) be a graph.
	Given \(U\subseteq V\) and \(v\in V\setminus N[U]\), we denote the connected component of \(G-N[U]\) containing \(v\) by \(\concomp[G]{v}{U}\).
	If \(G\) is clear from the context, we will omit the subscript.
\end{definition}

We define the winning condition (WC) as follows: 
there exist two edges claimed by Maker in the same connected component of \(\etgraph{G}-N[E_b^i]\) in some round \(i\), where \(E_b^i\) denote the edges claimed by Breaker up to this point.
By \cref{lem:maker-winning-structure}, (WC) is indeed a winning condition for Maker .

\begin{definition}
	Let \(G\) be Maker's win.
	The set \(M(G)\) is the set of Maker strategies that, given any configuration of the game graph, achieve (WC) in the fewest possible number of rounds.
	Once (WC) is achieved, the strategies in \(M(G)\) claim the edge-vertices along a shortest path between the two edge-vertices in the same connected component to win as in the proof of \cref{lem:maker-winning-structure}.
\end{definition}

The strategies in \(M(G)\) can be seen as an optimal winning strategies in the sense that the winning condition is achieved in the fewest number of turns.
However, this does not necessarily coincide with winning in the fewest number of turns.
Recall that the unbiased triangle game is defined such that it ends once Maker has managed to claim a triangle.
The next theorem shows that any strategy \(M\) from \(M(G)\) plays monotonically in the sense that it will always pick an edge from the connected component of the edge-triangle incidence graph of the current game graph that contains the last edge \(M\) claimed.
Intuitively, this makes sense as \(M\) tries to achieve (WC) as fast as possible and this can only be done by claiming edges that lie in the same connected component as some other edge previously claimed.

\begin{theorem}\label{thm:monotonicity}
	Let \(G\) be a graph such that Maker wins on \(\gamehg{G}\) and let \(M\in M(G)\).
	Then \(M\) satisfies the following property for all Breaker strategies \(B\):
	Denote the edge claimed by Maker (resp.~Breaker) in round \(i\) by \(e_i\) (resp.~\(a_i\)).
	Set \(E_b^i=\{a_1,\dots,a_i\}\).
	Let \(r\) be the number of rounds played.
	Then for all \(i\) such that \(1\leq i<r\) the inclusion \(\concomp[\etgraph{G}]{e_{i+1}}{E_b^{i+1}}\subseteq\concomp[\etgraph{G}]{e_i}{E_b^i}\) holds, i.e.~\(M\) plays monotonically.
\end{theorem}
\begin{proof}
	Let \(\ell\) be the round in which the strategy \(M\) first achieves the winning condition against \(B\).
	Given \(K\subseteq V(\etgraph{G})\) and \(E_1\subseteq E(G)\), denote by \(r(K,E_1)\) the maximum number of rounds \(M\) requires to achieve (WC) while only claiming edges from \(K\) given that \(M\) has already claimed the edges of \(K\cap E_1\), where the maximum is taken over all Breaker strategies.
	If no strategy for achieving (WC) exists under these constraints, \(r(K,E_1)\) is defined to be \(\infty\).
	We will omit the subscript in the connected component notation from now on.
	\begin{claim}\label{claim:round_minimizer}
		After each round \({1\leq i<\ell}\), the unique minimizer of \(r_i(K)\coloneqq r(K,\{e_j:1\leq j\leq i\}\) among the connected components of \(\etgraph{G}-N[E_b^i]\) is \(\concomp{e_i}{E_b^i}\).
	\end{claim}
	\begin{claimproof}
		Consider \(i=1\).
		Assume for contradiction that \(r_1(K)\leq r_1(\concomp{e_1}{E^1_b})\) for some connected component \(K\neq\concomp{e_1}{E^1_b}\) of \(\etgraph{G}-N[E_b^1]\).
		As Maker's first move was \(e_1\in\concomp{e_1}{E^1_b}\), Maker has no edge claimed in \(K\), so Maker could have claimed some edge in \(K\) instead of \(e_1\) in the first round to achieve (WC) in fewer rounds contradicting the definition of \(M(G)\).
		
		Let the statement be true for \(i<\ell-1\), we will prove it for \(i+1\).
		Let \(K\) be the connected component of \(\etgraph{G}-N[E_b^i]\) minimizing \(r_{i+1}\).
		Then \({r_{i+1}(K)<r_i(\concomp{e_i}{E_b^i})}\) as \(M\) achieves (WC) as fast as possible given the current configuration of the game and one round of play has passed.
		If \({K\cap\{e_1,\dots,e_{i+1}\}=\emptyset}\), then Maker has no edge claimed in \(K\) and could have claimed some edge in \(K\) in the first round to achieve (WC) faster, contradicting the definition of \(M(G)\).
		Therefore, \(K\cap\{e_1,\dots,e_{i+1}\}\neq\emptyset\), so \(K\) must be one of the connected components \(\concomp{e_1}{E_b^{i+1}},\dots,\concomp{e_{i+1}}{E_b^{i+1}}\).
		These connected components are disjoint as (WC) is first achieved in round \(\ell\).
		Observe that \(\concomp{e_j}{E_b^{i+1}}\subseteq\concomp{e_j}{E_b^i}\) for all \({1\leq j\leq i}\).
		As \(e_j\) is the only edge claimed by Maker in \(\concomp{e_j}{E_b^k}\) in round \(k\), where \({1\leq k<\ell}\), we conclude that \({r_i(\concomp{e_j}{E_b^i})\leq r_{i+1}(\concomp{e_j}{E_b^{i+1}}})\) for all \(1\leq j\leq i\).
		By the induction hypothesis, \(\concomp{e_i}{E_b^i}\) minimizes \(r_i\).
		This yields \(r_{i+1}(K)<r_i(\concomp{e_i}{E_b^i})\leq r_i(\concomp{e_j}{E_b^i})\leq r_{i+1}(\concomp{e_j}{E_b^{i+1}}\) for all \(1\leq j\leq i\), so \(K\) must be \(\concomp{e_{i+1}}{E_b^{i+1}}\).
	\end{claimproof}
	
	As \(M\) plays in such a way that (WC) is achieved the fastest, \cref{claim:round_minimizer} implies \({\concomp{e_{i+1}}{E_b^{i+1}}\subseteq\concomp{e_i}{E_b^i}}\) for all \(1\leq i\leq\ell\).
	If this was not the case, Maker would be able to achieve (WC) in \(k\) rounds against all possible Breaker strategies after making a move in a connected component of \(\etgraph{G}-N[E^i_b]\) in which he can only achieve (WC) within \(k+2\) or more turns as the minimizing component is unique.
	In round \(\ell\), (WC) is achieved and \(M\) switches to building a path from \(e_\ell\) to \(e_{\ell-1}\) (or vice versa).
	It is clear that the path-building strategy satisfies the monotonicity property.
\end{proof}

A simple consequence of \cref{thm:monotonicity} is that only the connected components of \(\etgraph{G}\) are relevant for Maker's win.
While \cref{cor:etg-connected} also follows from elementary properties of Maker-Breaker games, its proof becomes very simple when using monotonicity.

\begin{corollary}\label{cor:etg-connected}
	let \(G\) be a graph and \(K_1,\dots,K_\ell\) be the connected components of \(\etgraph{G}\).
	Then Maker wins on \(\gamehg{G}\) if and only if Maker wins on \(\gamehg{G[E(G)\cap K_i]}\) for some \(1\leq i\leq\ell\).
\end{corollary}
\begin{proof}
	Maker wins on \(\gamehg{G}\) if Maker wins on \(\gamehg{G[E\cap K_i]}\) as \(G[E\cap K_i]\) is a subgraph of \(G\).
	If Maker wins on \(\gamehg{G}\), \cref{thm:monotonicity} yields a monotone winning strategy for Maker.
	In particular, all edges claimed by this Maker strategy lie in the connected component \(K\) of \(\etgraph{G}\) containing the first edge \(e_1\) claimed by Maker.
	As this edge is independent of the strategy used by Breaker, Maker wins on \(\gamehg{G[E(G)\cap K]}\).
\end{proof}

From an algorithmic perspective, \cref{cor:etg-connected} says that it is sufficient to consider the connected components of \(\etgraph{G}\) separately to determine if \(G\) is Maker's win.
It also subsumes the trivial observation that only edges which lie on triangles are relevant for Maker's win as either all edges or no edges inside a connected component of \(\etgraph{G}\) lie on a triangle, the latter case leading to subgraphs consisting of a single isolated edge, which are obviously losing for Maker. 

Cycles play a central role for the strategies in \(M(G)\) as the remaining two results of this section show.

\begin{lemma}\label{lem:only-cycle-wc}
	Let \(G\) be a graph such that Maker wins on \(\gamehg{G}\) and let \(M\in M(G)\).
	Let \(\ell\) be the round in which \(M\) achieves (WC) against some Breaker strategy \(B\).
	For all \(1\leq i<\ell\), in round \(i+1\), \(M\) chooses an edge that lies on a cycle in \(\etgraph{G}-N[E_b^i]\).
\end{lemma}
\begin{proof}
	Assume for contradiction that there exists a Breaker strategy such that, when playing against this strategy, \(M\) chooses an edge \(e_{i+1}\) in round \(2\leq i+1\leq\ell\) that does not lie on a cycle in \(\etgraph{G}-N[E_b^i]\).
	Let \(B\) be a Breaker strategy with this property that maximizes the number of rounds \(M\) requires to achieve (WC) among such strategies.
	As before, we denote the edges claimed by Maker by \(e_1,\dots,e_\ell\) and the edges claimed by Breaker by \(a_1,\dots,a_\ell\).
	Let \(2\leq i+1\leq\ell\) be the first round in which \(M\) picks an edge that does not lie on a cycle in \(\etgraph{G}-N[E_b^i]\) against \(B\).
	We distinguish between the two cases \(i+1=\ell\) and \(i+1<\ell\).
	
	Assume first that \(i+1=\ell\).
	As \(e_\ell\) does not lie on a cycle, all paths from \(e_\ell\) to \(e_{\ell-1}\) in \(\concomp{e_{\ell-1}}{E_b^{\ell-1}}\) have the same second vertex, which is a triangle \(t\).
	Because (WC) is achieved in round \(\ell\), at most two of the edges used by the triangle \(t\) have been claimed by Maker, so Breaker can choose some edge to destroy the triangle \(t\), thereby removing all paths between \(e_\ell\) and \(e_{\ell-1}\).
	This is a contradiction to the choice of \(B\).
	
	Now assume \(i+1<\ell\).
	Without loss of generality, we may assume that the edge \(a_{i+1}\) is an edge \(a^*\) used by the first triangle-vertex of a path from \(e_{i+1}\) to \(e_i\) in \(\concomp{e_i}{E_b^i}\).
	Let \(a\) be another edge available to Breaker that destroys all paths from \(e_{i+1}\) to \(e_i\) in \(\concomp{e_i}{E_b^i}\).
	Then \(\concomp{e_{i+1}}{E_b^i\cup\{a^*\}}\subseteq\concomp{e_{i+1}}{E_b^i\cup\{a\}}\), so the choice \(a_{i+1}=a^*\) is optimal given that \(M\) plays monotonically, which is the case by \cref{thm:monotonicity}.
	Thus, this choice of \(a'_{i+1}\) is consistent with the maximality of the chosen strategy \(B\).
	
	We will define a strategy \(M'\) which will achieve (WC) faster against \(B\) contradicting the definition of \(M(G)\).
	The strategy \(M'\) plays like \(M\) in the first \(i\) rounds.
	In round \(i+1\), \(M'\) chooses \(e_{i+2}\) instead of \(e_{i+1}\) and then continues playing according to the strategy \(M\) in \(\concomp{e_{i+2}}{E_b^i\cup\{a'_{i+1}\}}\) where \(a'_{i+1}\) denotes the edge \(B\) claims in response to \(e_{i+2}\).
	
	\begin{figure}[t]
		\centering
		\includegraphics[scale=0.7]{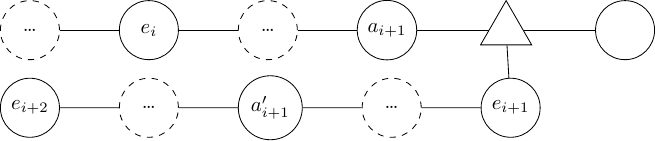}
		\caption{An exemplary illustration of \(\concomp{e_i}{E_b^i}\) in the second case of the proof of \cref{lem:only-cycle-wc}.
			This figure does not aim to cover all possibilities, e.g.~\(a_i\) could also be the unlabeled edge incident to the same triangle instead.}
		\label{fig:only-cycles-proof-1}
	\end{figure}
	
	We claim that \(M\) achieves (WC) in \(\concomp{e_{i+2}}{E_b^i\cup\{a'_{i+1}\}}\) in at most the same number of rounds it takes \(M\) to achieve (WC) in \(\concomp{e_{i+2}}{E_b^{i+2}}\) given that Maker has only claimed the edge \(e_{i+2}\) in these components.
	By \cref{thm:monotonicity}, \(e_{i+1},e_{i+2}\in\concomp{e_i}{E_b^i}\).
	Thus, \(B\) must destroy all paths from \(e_i\) to \(e_{i+2}\) in \(\concomp{e_i}{E_b^i}\) using \(N[a'_{i+1}]\) to prevent \(M'\) from achieving (WC) in round \(i+1\).
	As \(M\) achieves (WC) in round \(\ell\) and \({\concomp{e_{i+2}}{E_b^{i+2}}\subseteq\concomp{e_{i+1}}{E_b^{i+1}}\subseteq\concomp{e_i}{E_b^i}}\), \(B\) separates \(e_i\) from \(e_{i+1}\) with \(N[a_{i+1}]\) in \(\concomp{e_i}{E_b^i}\), but \(e_{i+2}\) remains reachable from \(e_{i+1}\) in \(\concomp{e_{i+1}}{E_b^{i+1}}\).
	Because \(a_{i+1}\) is an edge used by the first triangle-vertex of a path from \(e_{i+1}\) to \(e_i\) in \(\concomp{e_i}{E_b^i}\) and \(e_{i+1}\) does not lie on a cycle in \(\concomp{e_i}{E_b^i}\), all paths from \(e_{i+2}\) to \(e_i\) in \(\concomp{e_i}{E_b^i}\) must use \(e_{i+1}\).
	Therefore, \(a'_{i+1}\) must be chosen in such a way that \(N[a'_{i+1}]\) intersects a path from \(e_i\) to \(e_{i+1}\) in \(\concomp{e_i}{E_b^i}\) or contains an internal vertex of a path from \(e_{i+1}\) to \(e_{i+2}\) in \(\concomp{e_i}{E_b^i}\).
	In the first case, we have \({\concomp{e_{i+2}}{E_b^{i+2}}\subseteq\concomp{e_{i+2}}{E_b^i\cup\{a'_{i+1}\}}}\), which immediately implies the claim.
	
	Otherwise, \(N[a'_{i+1}]\) does not contain a vertex that lies on path from \(e_i\) to \(e_{i+1}\) in \(\concomp{e_i}{E_b^i}\), so \(a'_{i+1}\) must lie in \(\concomp{e_{i+1}}{E_b^{i+1}}\) and \(a'_{i+1}\) must be different from \(e_{i+1}\).
	This situation is visualized in \cref{fig:only-cycles-proof-1}.
	Then \(a'_{i+1}\) is also a possible choice for \(B\) in round \(i+2\) of the game against \(M\).
	Assume for contradiction that \(M\) achieves (WC) strictly faster in \(\concomp{e_{i+2}}{E_b^{i+2}}\) than in \(\concomp{e_{i+2}}{E_b^i\cup\{a'_{i+1}\}}\).
	Then the strategy \(B\) was not chosen maximally as claiming \(a'_{i+1}\) in round \(i+2\) instead of \(a_{i+2}\) forces \(M\) to continue playing in the connected component \(\concomp{e_{i+2}}{E_b^{i+1}\cup\{a'_{i+1}\}}=\concomp{e_{i+2}}{E_b^i\cup\{a'_{i+1}\}}\).
	By then following a strategy that delays (WC) in \(\concomp{e_{i+2}}{E_b^i\cup\{a'_{i+1}\}}\) as much as possible, Breaker prevents \(M\) from achieving (WC) in \(\ell\) rounds, which is a contradiction.
	
	Thus, \(M\) achieves (WC) in \(\concomp{e_{i+2}}{E_b^i\cup\{a'_{i+1}\}}\) in at most the number of rounds \(M\) requires to achieve (WC) in \(\concomp{e_{i+2}}{E_b^{i+2}}\).
	However, \(M'\) does not claim \(e_{i+1}\) in round \(i+1\) and instead directly claims \(e_{i+2}\), so it achieves (WC) at least one round faster against \(B\) than \(M\) contradicting the definition of \(M(G)\).
\end{proof}

\begin{lemma}\label{lem:only-cycle-first}
	Let \(G\) be a graph such that Maker wins on \(\gamehg{G}\).
	Then \(M(G)\) contains a strategy \(M\) whose first claimed edge lies on a cycle in \(\etgraph{G}\).
\end{lemma}
\begin{proof}
	Let \(M\in M(G)\) and consider the game of \(M\) against some Breaker strategy \(B\).
	As before, denote the edges claimed by Maker by \(e_1,\dots,e_\ell\) and the first two edges claimed by Breaker by \(a_1,a_2\).
	Assume that \(e_1\) does not lie on a cycle in \(\etgraph{G}\).
	Let \(a'_1\) be an edge used by the first triangle-vertex on a path \(P\) from \(e_1\) to \(e_2\) in \(\etgraph{G}\) that does not lie in \(V(=)\).
	Let \(e^*_1\) be the edge \(M\) picks if Breaker chooses \(a'_1\) in response to \(e_1\).
	By \cref{lem:only-cycle-wc}, \(e^*_1\) lies on a cycle.
	This situation is visualized in \cref{fig:only-cycles-proof-2}	

	We define the strategy \(M'\) which satisfies the desired property.
	The strategy \(M'\) first picks \(e^*_1\).
	Let \(a^*_1\) be Breaker's response.
	If \(a^*_1\notin\concomp{e^*_1}{\{a'_1\}}\), \(e^*_1\) lies on a cycle in \(\concomp{e^*_1}{\{a'_1\}}\).
	Then \(M'\) can achieve (WC) in the next round by claiming any available edge which lies on a shortest cycle containing \(e^*_1\) in \(\concomp{e^*_1}{\{a'_1\}}\).
	Else, \(a^*_1\) lies in \(\concomp{e^*_1}{\{a'_1\}}\).
	If \(e_1\notin\concomp{e^*_1}{\{a^*_1\}}\), \(M'\) plays according to \(M\) in \(\concomp{e^*_1}{\{a^*_1\}}\).
	Otherwise, \(M'\) picks \(e_2\) and plays according to \(M\) in \(\concomp{e_2}{\{a^*_1,a^*_2\}}\), where \(a^*_2\) denotes Breaker's response to \(e_2\).
	
	\begin{figure}[t]
		\centering
		\includegraphics[scale=0.7]{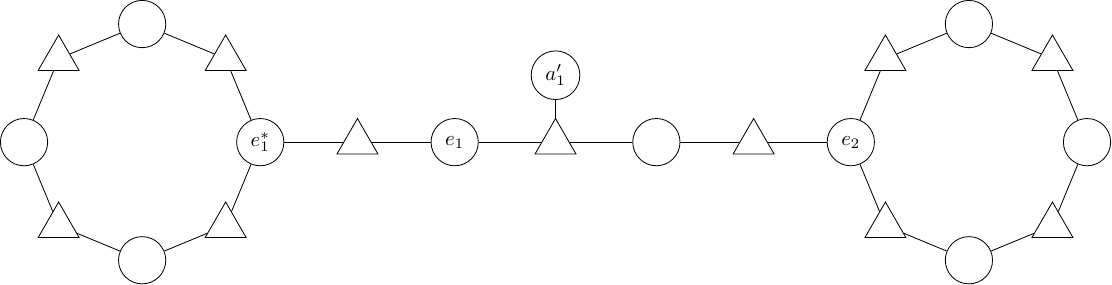}
		\caption{A concrete example of the situation in the proof of \cref{lem:only-cycle-first}.
			Except for \(a'_1\), edge-vertices with degree~1 are not drawn.}
		\label{fig:only-cycles-proof-2}
	\end{figure}
	
	Assume for contradiction that there exists a Breaker strategy \(B'\) that wins against \(M'\).
	It follows that \({a^*_1\in\concomp{e^*_1}{\{a'_1\}}}\) as otherwise \(M'\) achieves (WC) in the next round.
	In the case \(e_1\notin\concomp{e^*_1}{\{a^*_1\}}\), i.e.~\(N[a^*_1]\) separates \(e_1\) and \(e^*_1\) in \(\etgraph{G}\), \(M'\) wins in \(\concomp{e^*_1}{\{a^*_1\}}\).
	This can be seen as follows:
	In the configuration \((\{e_1\},\{a'_1\})\), \(M\) plays \(e^*_1\).
	If Breaker responds with \(a^*_1\), \(M\) continues to play in \(\concomp{e^*_1}{\{a'_1,a^*_1\}}\) by \cref{thm:monotonicity}.
	Because \(N[a^*_1]\) separates \(e_1\) and \(e^*_1\), we obtain \({\concomp{e^*_1}{\{a'_1,a^*_1\}}=\concomp{e^*_1}{\{a^*_1\}}}\), so \(M\) wins in \(\concomp{e^*_1}{\{a^*_1\}}\).
	Therefore, \(M'\) also wins in \(\concomp{e^*_1}{\{a^*_1\}}\).
	
	This leaves only the case where \(a^*_1\) lies in \(\concomp{e^*_1}{\{a'_1\}}\), but does not separate \(e^*_1\) and \(e_1\).
	In particular, we obtain \(e_2\in\concomp{e^*_1}{\{a^*_1\}}\).
	In this situation, \(M'\) picks \(e_2\) and \(B'\) responds with an edge \(a^*_2\).
	Then \(N[a^*_2]\) separates \(e^*_1\) and \(e_2\) as otherwise (WC) is already achieved.
	If \(N[a^*_2]\) does not contain an internal vertex of a path from \(e_1\) to \(e_2\) in \(\concomp{e_1}{\{a_1\}}\), \(M'\) wins in \(\concomp{e_2}{\{a^*_1,a^*_2\}}\supseteq\concomp{e_2}{\{a_1,a_2\}}\) as \(M\) wins in \(\concomp{e_2}{\{a_1,a_2\}}\).
	Thus, \(a^*_2\) must lie in \(\concomp{e_1}{\{a_1\}}\).
	As \(M\) plays monotonically and picks \(e_2\) in the configuration \((\{e_1\},\{a_1\})\), \(M\) can win in \(\concomp{e_2}{\{a_1,a^*_2\}}\subseteq\concomp{e_2}{\{a^*_1,a^*_2\}}\).
	Thus, \(M'\) wins against \(B'\), which is a contradiction.
	
	It remains to show that \(M'\) achieves (WC) in at most the same number of turns as \(M\).
	Let \(r\) be the maximum number of rounds \(M'\) requires to achieve (WC) and let \(B\) be a Breaker strategy against which \(M'\) requires \(r\) rounds to achieve (WC).
	Assume for contradiction that \(M\) achieves (WC) in at most \({r-1}\) rounds against any Breaker.
	In the game of \(M'\) against \(B\), it cannot be the case that \({a^*_1\notin\concomp{e^*_1}{\{a'_1\}}}\) because (WC) is achieved in two rounds in this situation.
	It is impossible to achieve (WC) without claiming two edges, so \(M\) cannot achieve (WC) faster than \(M'\) in this case.
	In the case \(e_2\in\concomp{e^*_1}{\{a^*_1\}}\), \(M\) and \(M'\) play identically except for the first round, so this is also impossible.
	Thus, \(N[a^*_1]\) must be a separator of \(e_1\) and \(e^*_1\).
	Consider the following Breaker strategy \(B'\) playing against \(M\).
	After \(M\) picks \(e_1\), \(B'\) picks \(a'_1\), to which \(M\) replies with \(e^*_1\).
	In response, \(B'\) plays \(a^*_1\) and continues playing according to \(B\) in \(\concomp{e^*_1}{\{a^*_1\}}\).
	Then \(M\) takes \(r+1\) rounds to win against \(B'\) which is contradicts the assumption that \(M\) achieves (WC) in at most \(r-1\) rounds against any Breaker.
	It follows that \(M'\in M(G)\).
\end{proof}

The analysis done in this section gives us a high-level view of winning strategies in \(M(G)\).
By \cref{lem:only-cycle-wc,lem:only-cycle-first}, they can be assumed to only claim edge-vertices \(e_i\) that lie on a cycle in \(\etgraph{G}-N[E_b^i]\) until (WC) is achieved.
These cycles can be seen as local threats.
Breaker must destroy all cycles containing \(e_i\) as otherwise Maker can achieve (WC) by claiming another edge-vertex that lies on a shortest cycle containing \(e_i\) in the next round.
Thus, Breaker must respond with an edge such that all cycles containing \(e_i\) also contain some triangle using this edge.
In this sense, the threat is local.
For example, if the last edge claimed by Maker lies on two cycles \(C_1\) and \(C_2\), then Breaker must respond with an edge from \(N(\triasCap{C_1})\cap N(\triasCap{C_2})\).
By \cref{thm:monotonicity}, the strategies are monotone, so after the first round, there is also a global threat caused by the edge claimed by Maker in the previous round.
To prevent Maker from achieving (WC), Breaker must destroy all paths between \(e_{i-1}\) and \(e_i\).
Assuming Breaker was successful in destroying all cycles containing \(e_{i-1}\) in the previous round, this is trivial on its own, but this must be achieved while also dealing with the local threats.

From an algorithmic perspective, it is enough to only consider strategies for Maker that satisfy \cref{thm:monotonicity,lem:only-cycle-wc,lem:only-cycle-first} to establish or refute that a given graph is Maker's win.
These results can also be useful for proving that some graph is not Maker's win as it is enough to show that Breaker can win against all monotone strategies to establish that Maker does not win.
This is used in the proof of \cref{thm:cactus-game-chara}.

\section{Algorithmic Results}\label{sec:complexity}
In this section, we analyze the complexity of the unbiased triangle game, determining a lower bound as well as upper bounds for three classes of graphs.
While the unbiased Maker-Breaker \(H\)-game, where Maker tries to claim an edge-induced copy of \(H\) in the input graph, is \pspace-complete, even when Maker is trying to claim a fixed tree~\cite{DBLP:journals/dam/DucheneGINOPS25}, the outcome of the game can be determined in polynomial time if Maker is trying to claim a graph consisting of at most~3 edges~\cite{galliot2025makerbreakersolvedpolynomialtime}.
Specifically, Galliot, Gravier and Sivignon show that the winner of the unbiased Maker-Breaker game played on a hypergraph \(H\) of rank at most~3 can be determined in time \(\Oh{|V(H)|^5|E(H)|^2+|V(H)|^6\Delta(H)}\) where \(\Delta(H)\) denotes the maximum degree of a vertex in \(H\).
This is done by simulating three rounds of the game and checking if Maker was able to claim at least one of two structures indicating that he will win under optimal play.
The three rounds of play correspond to the summand \(|V(H)|^6\Delta(H)\) as up to six vertices will be claimed in three rounds.
The other summand, \(|V(H)|^5|E(H)|^2\), corresponds to a precomputation that simplifies checking the existence of the structures indicating Maker's win.

When modeling the unbiased triangle game played on a graph \(G\) with a hypergraph \(H\), the vertices of \(H\) are the edges of \(G\) and the hyperedges of \(H\) are those sets of exactly three edges that form a triangle in \(G\).
As any edge participates in at most \(n\) triangles, \(\Delta(H)\leq n\).
Thus, the bound \({\Oh{|V(H)|^5|E(H)|^2+|V(H)|^6\Delta(H)}}\) becomes \(\Oh{m^5t^2+nm^6}\).
This does not yet include the time required to compute the game hypergraph from the board graph which can be bounded by \(\Oh{n+m^{1.5}}\), see \cref{lem:computing-etg}.
Therefore, the worst-case running time of \cite{galliot2025makerbreakersolvedpolynomialtime} is \(\Oh{n^{16}}\) for sufficiently dense graphs, i.e.~when \(m=\Omega(n^2)\) and \({t=\Omega(n^3)}\).
Our goal is to decide the unbiased triangle game significantly faster.

We improve the bound on the decision complexity to \(\Oh{\generalTime}\).
However, this still leaves room for further improvements.
The rationale of our approach towards an optimal time complexity, which we conjecture to be \(\Oh{n^3}\), is the following:
We first largely improve the decision complexity when the board graph contains \(K_4\) as a subgraph and the edge-triangle incidence graph is connected.
This may open the path to tackle the situation where the board graph is \(K_4\)-free using the field of extremal graph theory. 
For example, it is not difficult to see that the Turán graph with three color classes and seven vertices \(T_3(7)\) is Maker's win.
For the class of graphs whose edge-triangle incidence graph is a cactus graph, which is a subclass of the class of \(K_4\)-free graphs, we prove a decision complexity which cannot be improved by an algorithm that explicitly computes the edge-triangle incidence graph.
For other \(K_4\)-free graphs, the hope is to succeed as they have a large number of triangles, see e.g.~\cite{DBLP:journals/jct/ChudnovskyRST10,he2025numbertrianglesk4freegraphs}.

\textbf{Deciding the game on general graphs.}
The unbiased triangle game is a special case of the Maker-Breaker game on hypergraphs of rank at most~3.
This suggests that the running time can be improved by exploiting the particular structure of the unbiased triangle game.
Indeed, we show that this is case.
A major difference between general rank~3 hypergraphs and those induced by the unbiased triangle game is that different triangles can intersect in at most one edge, so the hypergraph is linear.
The Maker-Breaker game on linear rank~3 hypergraphs has already been investigated by Kutz~\cite{kutz2004thesis}, who gave a characterization of Breaker's win through the structure of the hypergraph if Maker has already claimed a vertex.
Kutz only roughly analyzes the complexity of the algorithm for deciding the outcome of the game implied by his characterization, arguing that it runs in polynomial time without giving a concrete running time bound.
Our more detailed analysis shows that Kutz' characterization can be used to determine the outcome of the Maker-Breaker game played on a linear hypergraph \(H\) of rank at most~3 in time \(\Oh{|V(H)|^3+|V(H)|^2|E(H)|}\), see \cref{thm:general-alg}.
In fact, when \(H\) is connected and not~3-uniform, a running time of \(\Oh{|V(H)|+|E(H)|}\) can be achieved, see \cref{cor:kutz-non-3-uniform}.
For the unbiased triangle game, this corresponds to a running time of \(\Oh{\generalTime}\), which can be of order \(\Oh{n^7}\) in the worst case.
This is a major improvement on the \(\Oh{n^{16}}\) running time of~\cite{galliot2025makerbreakersolvedpolynomialtime}.

\textbf{Deciding the game on graphs containing \(\mathbf{K_4}\).}
Another difference is that the unbiased triangle game induces linear hypergraphs that cannot contain arbitrary cycles of length~3.
Any cycle of length~3 is created by a copy of \(K_4\) in the board graph \(G\), so a cycle of length~3 can only occur as part of the hypergraph-equivalent of \(\etgraph{K_4}\).
Due to this special role of \(K_4\), we are able to obtain a bound of \(\Oh{\cliqueTime}\) for board graphs \(G\) such that \(K_4\subseteq G\) and \(\etgraph{G}\) is connected.
The presence of a \(K_4\) allows us to either find that the board graph is Maker's win or reduce to at most seven Maker-Breaker games in related hypergraphs which are not~3-uniform, see \cref{thm:k4-elimination} for details.
This allows us to use \cref{cor:kutz-non-3-uniform} to obtain the improved running time.
The summand \(\cliqueTimeAddition\) corresponds to the time spent identifying a subgraph isomorphic to \(K_4\).
If it is not just guaranteed that \(G\) contains \(K_4\) as a subgraph, but a vertex set \(U\) such that \(G[U]\cong K_4\) is also given as a witness, the running time improves to \(\Oh{\cliqueTimeBase}\).

\textbf{Deciding the game when \boldmath\(\etgraph{G}\) is a cactus graph.}
The best upper bound is achieved for graphs \(G\) such that \(\etgraph{G}\) is a cactus graph.
In this situation, Maker's win is characterized by a condition that can be checked in the edge-triangle incidence graph without requiring the simulation of any rounds of play.
The algorithm implied by our winning condition runs in time \(\Oh{\cactusTime}\).
This can be seen as the first step towards deciding the unbiased triangle game on \(K_4\)-free graphs.
Observe that \(\etgraph{K_4}\) is not a cactus graph, see~\cref{fig:etgraph-example}.
Thus, the board graph must be \(K_4\)-free if \(\etgraph{G}\) is a cactus graph.

Our results can be interpreted in the following way:
If the input graph is sufficiently dense, one would expect that it contains \(K_4\) as a subgraph and that \(\etgraph{G}\) is connected.
On the other hand, if the input is sufficiently sparse, one would expect \(\etgraph{G}\) to be a cactus graph.
For these extremes of the spectrum, faster algorithms are available.
The remaining part of the spectrum, i.e.~graphs which are neither particularly dense nor sparse, are covered by our general algorithm with worse running time.

\textbf{Lower bound.}
The lower bound follows from a linear-time reduction of triangle detection in a graph \(G\) to the outcome of the unbiased triangle game on a graph \(H\), see \cref{thm:trig-detection-lb}.

\subsection{Maker-Breaker Games on Linear Hypergraph of Rank Three}
In this section, we show that the outcome of the unbiased Maker-Breaker game played on a linear hypergraph \(H\) of rank at most~3 can be decided in time \(\Oh{|V(H)|^3+|V(H)|^2|E(H)|}\).
This highly improves on the general algorithm of~\cite{galliot2025makerbreakersolvedpolynomialtime} with running time \({\Oh{|V(H)|^5|E(H)|^2+|V(H)|^6\Delta(H)}}\).
We employ some results of Kutz~\cite{kutz2004thesis} in the construction our algorithm.

\begin{definition}[\cite{kutz2004thesis}]
	A~2-edge in a hypergraph is a hyperedge of size exactly~2.
	A hypergraph is called almost~3-uniform if all its hyperedges have size exactly~3 except for one hyperedge which has size exactly~2.
	In particular, an almost~3-uniform hypergraph must contain a~2-edge.
\end{definition}

\begin{definition}[\cite{kutz2004thesis}]\label{defn:dock}
	Let \(H\) be an almost~3-uniform hypergraph with 2-edge \(\{\alpha,\beta\}\in E(H)\).
	A set of two vertices \({U=\{x,y\}\subseteq V(H)}\) is called a dock if \(U\cup\{\alpha\}\in E(H)\) or \(U\cup\{\beta\}\in E(H)\).
	If \(U\cup\{\alpha\}\in E(H)\), \(U\) is a lower dock, otherwise it is an upper dock.
	The collection of lower (resp.~upper) docks is the lower (resp.~upper) shore.
	The vertices \(x\) and \(y\) are called dock vertices.
	A dock \(\{x,y\}\) is closed if there is a path between \(x\) and \(y\) in \(H\) that uses neither \(\alpha\) nor \(\beta\).
	Otherwise, the dock is called open.
\end{definition}

The distinction between lower and upper docks is arbitrary as the role of \(\alpha\) and \(\beta\) is symmetric.
If Breaker wins on a hypergraph, vertices from different shores can only interact in a very restricted manner, see \cref{thm:kutz-breaker-chara}.
It is not important which one of the vertices \(\alpha\) and \(\beta\) is associated to the upper and the lower shore.
Instead, we are interested in the interaction between the shores.
The names lower and upper shore were used in \cite{kutz2004thesis} due to the way the author depicted the structure of the hypergraph, with \(\alpha\) always being at the bottom and \(\beta\) being at the top, see e.g. \cite[Figure~7]{kutz2004thesis}.

\begin{definition}[{\cite[Definition~33]{kutz2004thesis}}]\label{defn:kutz-ladder}
	A graph \(L\) is called a ladder of height \(h\geq0\) on \(a_0\) and \(c_0\), where \({a_0,c_0\in V(L)}\), if it can be constructed by the following procedure:
	\begin{itemize}
		\item begin with the graph \(L_0=(\{a_0,c_0\},\emptyset)\) consisting of two vertices \(a_0,c_0\)
		\item \textbf{for} \(i=1,\dots,h\) \textbf{do} (if \(h=0\), simply skip the loop)
			\begin{itemize}
				\item take a new path \(P_i=v_1^{(i)}\dots v_\ell^{(i)}\) with \(\ell\geq5\), where \(v_1^i=c_{i-1}\), \(v_\ell^i=a_{i-1}\) and no further vertices are common with \(L_{i-1}\)
				\item let \(a_i\coloneqq v_{\ell-2}^{(i)}\) and add a new vertex \(c_i\) as well as the edge \(c_iv_{\ell-1}^{(i)}\)
				\item set \(L_i\coloneqq L_{i-1}\cup P_i\cup Q_i\)
			\end{itemize}
		\item \textbf{either} end the construction of \(L\) by setting \(L\coloneqq L_h-c_h\)\\
		\textbf{or} take an optional additional path \(R\) from \(c_h\) to some vertex \(r\) of the path \(P_h\) except \(a_h\) (but \(r=c_{h-1}\) allowed) that contains no further vertices of \(L_h\) and let \(L\coloneqq L\cup R\).\qedhere
	\end{itemize}
\end{definition}

Ladders were defined slightly differently in the context of hypergraphs in~\cite{kutz2004thesis}.
We have rephrased them in terms of graphs because we want to remain within the framework of the vertex-hyperedge incidence graph.
This allows us to use classical graph algorithms to check structural properties such as 2-connectedness.
The following theorem is obtained by interpreting Kutz' characterization of Breaker's win~\cite[Theorem~38]{kutz2004thesis} through the lens of the vertex-hyperedge incidence graph.

\begin{theoremq}\label{thm:kutz-breaker-chara}
	Let \(H\) be a linear almost~3-uniform hypergraph such that \(V(\vegraph{H})=N[B]\) where \(B\) is a maximal 2-connected component of size at least~3 in \(\vegraph{H}\).
	Let \(e_{\alpha\beta}=\{\alpha,\beta\}\in E(H)\) be the unique~2-edge and assume that \(\alpha\) and \(\beta\) have degree at least~2.
	Denote the docks of \(H\) by \(\{x_1,y_1\},\dots,\{x_k,y_k\}\).
	Let \(G'\) be obtained from \(\vegraph{H}\) by removing all vertices of degree~1 and set \(G=G'-N[e_{\alpha\beta}]\).
 	Then Breaker wins on \(H\) if and only if \(G\) has the following structure:
	\begin{enumerate}[label=(\arabic*)]
		\item\label{item:kutz-chara-1} No two docks from the same shore are in the same connected component.
		Each closed dock \(D_1=\{x_i,y_i\}\) is connected to exactly one other dock \(D_2=\{x_j,y_j\}\) from the opposite shore, i.e.~there exists a path from \(s\) to \(t\) for some \(s\in D_1\) and \(t\in D_2\).
		For an open dock \(D_3\), each dock vertex \(v\in D_3\cap V(G)\) is connected to exactly one other dock from the opposite shore.
		
		\item\label{item:kutz-chara-2} If \(v_1\) and \(v_2\) are dock vertices from different open docks and they lie in the same connected component \(K\) of \(G\), then \(K\) is a path between \(v_1\) and \(v_2\).
		
		\item\label{item:kutz-chara-3} If \(D_1\) and \(D_2\) are both closed docks, then the connected component of \(G\) containing the docks consist of three paths, an \((x_i,y_i)\)-path \(P_1\), an \((x_j,y_j)\)-path \(P_2\) and a \((e_1,e_2)\)-path \(P_3\) where \(e_1\in\edgeCap{P_1}\) and \(e_2\in\edgeCap{P_2}\) such that \(|V(P_1)\cap V(P_3)|=|V(P_2)\cap V(P_3)|=1\) and \(V(P_1)\cap V(P_2)=\emptyset\).
		The path \(P_3\) consists of at least~3 vertices.
		
		\item\label{item:kutz-chara-4} If one dock is closed and the other is open, say \(D_1\) is closed and \(D_2\) is open, let \(K\) be the connected component containing \(D_1\) in \(G\).
		Let \(F=G'[K\cup\{\alpha,\beta\}]\) be \(K\) augmented by the vertices \(\alpha\) and \(\beta\).
		If \(D_1\) is an upper dock, then \(F\) is a ladder on \(a_0=\alpha\) and \(c_0=\beta\).
		If \(D_1\) is a lower dock, then \(F\) is a ladder on \(a_0=\beta\) and \(c_0=\alpha\).
		The ladder has height at least~1 and height at least~2 if the optional path is not present.\qedhere
	\end{enumerate}	
\end{theoremq}

Our first algorithmic result is \cref{thm:kutz-analysis}, which gives an algorithm for checking the conditions of \cref{thm:kutz-breaker-chara} in linear time with respect to the size of the hypergraph.

\begin{theorem}\label{thm:kutz-analysis}
	In the setting of \cref{thm:kutz-breaker-chara}, conditions~\labelcref{item:kutz-chara-1} to~\labelcref{item:kutz-chara-4} can be checked in time \(\Oh{|V(H)|+|E(H)|}\).
	Additionally, \(\vegraph{H}\) can be computed from \(H\) in the same running time.
\end{theorem}
\begin{proof}
	Clearly, \(|V(\vegraph{H})|=|V(H)|+|E(H)|\).
 	The number of edges in \(\vegraph{H}\) is bounded by thrice the number of edges in \(H\), so \(|E(\vegraph{H})|\leq 3|V(\vegraph{H})|\).
 	We will therefore analyze the complexity of the algorithm with respect to \(|\vegraph{H}|\).
 	The graph \(\vegraph{H}\) can be computed in time \(\Oh{|V(\vegraph{H})|}\) by starting from the graph \((V(H)\cup E(H),\emptyset)\) and then iterating over \(E(H)\) to add the edges according to the incidence structure.
	Then the graph \(G'\) is constructed by iterating over \(E(\vegraph{H})\) to compute the degree of all vertices and then iterating over \(V(\vegraph{H})\) to remove those vertices of degree~1.
	Finally, compute \(G=G'-N[e_{\alpha\beta}]\).
	
	Condition \labelcref{item:kutz-chara-1}: Iterate over \(E(H)\) to determine the docks, classifying them into lower and upper docks.
	Each dock of the upper shore can be classified as open or closed by using breadth-first search as follows:
	Mark all vertices of \(G\) as unvisited.
	Iterate over all docks \(\{x,y\}\) of the upper shore.
	If \(\{x,y\}\) is an open dock, only one of the vertices is guaranteed to be a vertex of \(G\).
	For vertices in \(V(\vegraph{H})\setminus V(G)\), the operations described in the following are skipped.
	If one of the dock vertices is already marked as visited, a previous dock vertex of the same shore is able to reach it, so condition~\labelcref{item:kutz-chara-1} is not satisfied.
	Otherwise, start a breadth-first search from \(x\).
	During the search, all reached vertices will be marked as visited. 
	If \(y\) is marked as visited after the search, the dock is closed.
	If \(y\) is not marked, execute a breadth-first search from \(y\) and mark the visited vertices.
	Repeat this process for the lower shore.
	Now every connected component of \(G\) intersects each shore in at most one dock.
	Execute another breadth-first search from each dock vertex of each shore to check that the dock vertex can reach a dock vertex from the appositive shore.
	If this is the case, condition \labelcref{item:kutz-chara-1} is satisfied.
	Otherwise, \labelcref{item:kutz-chara-1} is not satisfied.
	This process executes breadth-first search at most six times in each connected component of \(G\).
	It remains to check conditions~\labelcref{item:kutz-chara-2},~\labelcref{item:kutz-chara-3} and~\labelcref{item:kutz-chara-4}.
	
	Execute another breadth-first search from each dock vertex of an upper dock to determine the pairs of dock vertices (open dock to open dock), the pairs of docks (closed dock to closed dock) and the pairs of docks and dock vertices (closed dock to open dock) between which a connection exists.
	This searches each connected component of \(G\) at most twice as there are at most two dock vertices from the upper shore in any connected component of \(G\).

	Condition \labelcref{item:kutz-chara-2}: For a pair \((v_1,v_2)\) of dock vertices, another breadth-first search starting from \(v_1\) is used to check that the connected component containing the dock vertices is a path.
	Every vertex in the BFS-tree must have exactly one child, except for the last vertex, which must be \(v_2\).
	If this is not the case, condition~\labelcref{item:kutz-chara-2} is violated.
	
	Condition \labelcref{item:kutz-chara-3}: For a pair of docks \((D_1,D_2)\) with \(D_1=\{x_i,y_i\}\) and \(D_2=\{x_j,y_j\}\), \(i\neq j\), \(i,j\in\{1,\dots,k\}\), execute a breadth-first search from \(x_i\) and \(x_j\) to determine a shortest path \(P_1\) from \(x_i\) to \(y_i\) and a shortest path \(P_2\) from \(x_j\) to \(y_j\) in \(G\).
	Check that \(x_i,x_j,y_i,y_j\) all have degree~1 and that both paths contain exactly one hyperedge-vertex of degree~3.
	Mark all the vertices in \(V(P_1)\), then mark all the vertices in \(V(P_2)\).
	If a vertex was marked twice, the paths are not disjoint and condition~\labelcref{item:kutz-chara-3} is violated.
	Unmark the two hyperedge-vertices \(e_1\) and \(e_2\) of degree~3 and start a BFS from \(e_1\).
	The search only considers unmarked vertices.
	All internal vertices in the resulting tree must have exactly one child.
	The leaf must be \(e_2\).
	
	Condition \labelcref{item:kutz-chara-4}: For a pair \((D,v)\) with a dock \(D=\{x_i,y_i\}\) and a dock vertex \(v\), a slight modification of breadth-first search is required.
	The difficulty in detecting a ladder is caused by the optional path that may or may not be present in a ladder, see \cref{defn:kutz-ladder}.
	A vertex of degree~3 either connects \(a_i\), \(a_{i+1}\) and \(c_{i+1}\) or it is the end point \(r\) of the optional path.
	First, compute \({F=G'[K\cup\{\alpha,\beta\}]}\) where \(K\) is the connected component containing \(D\) in \(G\).
	Determine the vertices \(a_0\) and \(c_0\) based on which shore \(D\) belongs to.
	We now state our algorithm for detecting ladders.
	\begin{enumerate}[label=\arabic*, widest=11]
		\item[]\textbf{AlgorithmLadder}
		\item[]\textbf{Input}: A non-empty graph \(F\) and different vertices \(a_0,c_0\in V(F)\)
		\item[]\textbf{Output}: ``Input is a ladder'' or ``Input is not a ladder''. In the first case, the height of the ladder and whether or not the optional path is used by the ladder is also given.
		\item\label{alg:ladder-height-0} \(F\) is a ladder with height~0 if and only if \(V(F)=\{a_0,c_0\}\) and \(E(F)=\emptyset\).
		\item Initialize \(h\) to 0.
		\item\label{alg:ladder-while} \textbf{while}(true)
		\begin{enumerate}[label=\arabic*, widest=11]\setcounter{enumii}{\value{enumi}}
			\item\label{alg:ladder-goto-1} Execute a BFS from \(c_h\) until \(a_h\) is found.
			If \(a_h\) is never found, \(F\) is not a ladder.
			\item\label{alg:ladder-path-defn} Backtrack along the path \(P_{h+1}\) from \(a_h\) to \(c_h\) to check that \(a_h\) and \(c_h\) have degree~1, at most~2 vertices of the path have degree at least~3 and these vertices have degree exactly~3.
			If any of these conditions is violated, \(F\) is not a ladder.
			\item Determine the number of vertices \(\ell\) of the path from \(a_h\) to \(c_h\).
			\item\label{alg:ladder-no-sibling} If \(a_h\) has no sibling in the BFS tree, \(F\) is a ladder of height \(h+1\) if and only if \(\ell\geq5\) and no vertex of degree~3 was found.
			\item If \(a_h\) has a sibling in the BFS tree, set \(c_{h+1}\) to this vertex and set \(a_{h+1}\) to the grandparent of \(a_h\).
			\item\label{alg:ladder-degree-check} Check that there are exactly two vertices of degree~3 in the path, otherwise \(F\) is not a ladder.
			\item\label{alg:ladder-grandparent} If the grandparent of \(a_h\) does not have degree~3, exit the loop.
			\item\label{alg:ladder-length-check} If \(\ell<4\), \(F\) is not a ladder.
			\item\label{alg:ladder-goto-2} Remove from \(F\) the edges from \(a_{h+1}\) and \(c_{h+1}\) to their respective parent as well as the edge from the parent of \(a_h\) to \(a_{h+1}\), set \(h\coloneqq h+1\).
		\end{enumerate}\setcounter{enumi}{\value{enumii}}
		\item\label{alg:ladder-optional-path} The grandparent of \(c_h\) does not have degree~3 or \(\ell<4\), so the optional path must be used.
		Then \(F\) is a ladder of height \(h+1\) using the optional path if and only if the connected component of \(c_h\) contains exactly two vertices with degree at least~3, all vertices except \(a_h\) and \(c_h\) have degree at least~2, \(c_{h+1}\) has degree~2 and there exists a path consisting of at least~5 vertices from \(c_h\) to \(a_h\).
		Execute a final BFS from \(c_h\) to check this.
	\end{enumerate}
	By removing the edges in step~\labelcref{alg:ladder-goto-2} of the algorithm, we ensure that each vertex will be inspected in at most four of the breadth-first searches used as subroutines.
	This results in a linear running time although a non-constant number of breath-first searches may be executed.
	We now show correctness of the algorithm.
	
	The while-loop in step~\labelcref{alg:ladder-while} can be seen as a recursive execution of the algorithm.
	Therefore, we use induction over the height \(h\) of the ladder to prove that the algorithm accepts \(F\) if it is a ladder of height at least~1 on \(a_0\) and \(c_0\).
	Clearly, step~\labelcref{alg:ladder-height-0} of the algorithm correctly classifies ladders with height~0.
	Consider the case \(h=1\).
	If \(F\) does not use the optional path, it is simply a \((c_0,a_0)\)-path consisting of at least~5 vertices, so the algorithm will accept in step~\labelcref{alg:ladder-no-sibling}.
	If \(F\) does use the optional path, then any \((c_0,a_0)\)-path contains a vertex of degree~3, namely the predecessor of \(a_0\).
	The optional path may be shorter than the path \(P_1\) from \(c_0\) to \(a_0\) in \cref{defn:kutz-ladder}, so that it will be found by the breadth-first search in step~\labelcref{alg:ladder-goto-1}.
	Then the vertex \(v\) that the algorithm deems to be \(a_1\) is actually \(c_1\).
	However, this is irrelevant because both \(a_1\) and \(c_1\) have degree at most~2 if the optional path is used.
	Therefore, the algorithm will go to step~\labelcref{alg:ladder-optional-path} in step~\labelcref{alg:ladder-grandparent} by exiting the loop.
	Then, the algorithm will accept in step~\labelcref{alg:ladder-optional-path}.
	
	\(h\to h+1\): Let \(F\) be a ladder of height \(h+1\) on \(a_0\) and \(c_0\).
	Then \(F'=F-\left(V(P_1)\setminus\{a_0,c_0\}\right)\), with \(P_1\) as in \cref{defn:kutz-ladder}, is a ladder of height \(h\) on \(a_1\) and \(c_1\).
	The algorithm will accept \(F'\) by the induction hypothesis.
	The graph \(F'\) is exactly the maximal connected subgraph of \(F\) containing \(c_1\) after the removal of the edges in step~\labelcref{alg:ladder-goto-2} in the first iteration of the loop, so the algorithm accepts the input.
	
	It remains to show that the algorithm only accepts ladders.
	Let \(h\) be the height of the supposed ladder according to the algorithm.
	For every \(1\leq i<h\), the path \(P_i\) in step~\labelcref{alg:ladder-path-defn} consists of at least~5 vertices by step~\labelcref{alg:ladder-length-check} and contains exactly two vertices of degree at least~3 and these vertices have degree exactly~3 by the check in step~\labelcref{alg:ladder-degree-check}.
	For all \(1\leq i<h\), \(a_{i+1}\) is the grandparent of \(a_i\) and \(c_{i+1}\) is the sibling of \(a_i\) along \(P_i\).
	Additionally, the paths \(P_1,\dots,P_h\) are pairwise disjoint.
	Thus, the paths \(P_1,\dots,P_{h-1}\) are as in~\cref{defn:kutz-ladder}.
	If the algorithm returns that the supposed ladder does not use the optional path, the internal vertices of the path \(P_h\) all have degree~2 by step~\labelcref{alg:ladder-no-sibling}.
	Then \(P_h\) also satisfies the definition of a ladder.
	Therefore, the input is indeed a ladder that does not use the optional path.
	
	This leaves the case where the algorithm claims that the optional path is used.
	In this case, \(P_h\) also contains exactly two vertices of degree at least~3, having degree exactly~3, by step~\labelcref{alg:ladder-degree-check}.
	Thus, the two vertices of degree at least~3 in the connected component \(K\) containing \(c_h\) in step~\labelcref{alg:ladder-optional-path} are the two degree~3 vertices in \(P_h\).
	Every vertex in \(K\setminus\{a_h,c_h\}\) has degree at least~2, so \(K\) is the union of two \((a_h,c_h)\)-paths, one of which consists of at least~5 vertices.
	Both \(a_{h+1}\) and \(c_{h+1}\) do not have degree~3 by step~\labelcref{alg:ladder-grandparent} and step~\labelcref{alg:ladder-optional-path}.
	Therefore, the input is indeed a ladder using the optional path.

	Thus, the conditions of \cref{thm:kutz-breaker-chara} can be checked by a constant number of linear-time searches in each (augmented) connected component of \(G\) resulting in the running time of \(\Oh{|V(H)|+|E(H)|}\).
\end{proof}

\Cref{thm:kutz-breaker-chara} requires that \(\vegraph{H}\) consists of a single~2-connected component, potentially with some leafs attached to it.
To apply this result to a connected almost~3-uniform hypergraph \(H\) without this structure, \(\vegraph{H}\) is first decomposed into 2-connected components.
This is achieved by applying \cref{lem:cutpoint-split} and \cref{cor:lowdeg-pruning}, resulting in \cref{thm:block-split}.

\begin{lemmaq}[{\cite[Lemma~25]{kutz2004thesis}}]\label{lem:cutpoint-split}
	Let \(H\) be a connected linear almost~3-uniform hypergraph with~2-edge \({e\in E(H)}\).
	Let \(H=A\cup B\) be a decomposition of \(H\) with \(E(A)\neq\emptyset\neq E(B)\), \(|V(A)\cap V(B)|=1\) and \(e\in E(A)\).
	Let \(B_1,\dots,B_k\) be the connected components of \(B^{+v}\).
	Then each of the connected hypergraphs \(A,B_1,\dots,B_k\) contains at least one~2-edge and \(H\) is Maker's win if and only if at least one of the hypergraphs \(A,B_1,\dots,B_k\) is Maker's win.
\end{lemmaq}

\begin{corollary}\label{cor:lowdeg-pruning}
	Let \(H\) be a connected linear almost~3-uniform hypergraph and let \(e\in E(H)\) be a hyperedge of size three containing two vertices \(v_1\) and \(v_2\) of degree~1.
	Then Maker wins on \(H\) if and only if Maker wins on \({H-\{v_1,v_2\}}\).
\end{corollary}
\begin{proof}
	If Breaker wins on \(H\), Breaker also wins on \(H-\{v_1,v_2\}\), so assume \(H\) is Maker's win.
	The vertices \(v_1\) and \(v_2\) both have degree~1 in \(H\), so only the last vertex \(v_3\) of the edge \(e\) can be used by other edges of \(H\).
	Thus, \(H\) can be decomposed into \(H_1=H-\{v_1,v_2\}\) and \(H_2=(\{v_1,v_2,v_3\},\{e\})\) with \(V(H_1)\cap V(H_2)=\{v_3\}\).
	Breaker clearly wins on a hypergraph consisting of a single~2-edge, so Breaker wins on \(H_2^{+v_3}\).
	By \cref{lem:cutpoint-split}, Maker wins on \(H\) if and only if Maker wins on \(H_1\).
\end{proof}

Applying \cref{cor:lowdeg-pruning} until this is no longer possible ensures that each hyperedge-vertex has two neighbors with degree at least~2.
Then every hyperedge-vertex lies on a cycle or connects two different cycles.
This property is used in the proof of \cref{thm:block-split}.
\Cref{cor:lowdeg-pruning} can also be used to preprocess the board graph by eliminating edges that cannot contribute to Maker's win.

\Cref{thm:block-split} allows us to split the hypergraph into independent subproblems using the~2-connected components of the vertex-hyperedge incidence graph.
This avoids the iterative computation of suitable decompositions for \cref{lem:cutpoint-split} by exhaustively applying \cref{lem:cutpoint-split} in a single step.

\begin{theorem}\label{thm:block-split}
	Let \(H\) be a connected linear almost~3-uniform hypergraph and let \(e\in E(H)\) be the unique~2-edge.
	Let \(B_1,\dots,B_k\) be the maximal 2-connected components of size at least~3 of \(\vegraph{H}\) and set \({V_i=V(H)\cap N[B_i]}\).
	Then Maker wins on \(H\) if and only if one of the following conditions is satisfied:
	\begin{enumerate}[label=(\arabic*)]
		\item\label{item:thm:block-split-1} There exists \(1\leq i\leq k\) such that \(e\in B_i\) and Maker wins on \(H[V_i]\).
		\item\label{item:thm:block-split-2} There exists \(1\leq i\leq k\) such that \(e\notin B_i\) and the last vertex on a shortest path from \(e\) to \(B_i\) in \(\vegraph{H}\) is an element of \(V(H)\).
		\item\label{item:thm:block-split-3} There exists \(1\leq i\leq k\) such that \(e\notin B_i\) and Maker wins on \(H[V_i]^{+v}\) where \(v\in V(H)\) in \(\vegraph{H}\) is the penultimate vertex on a shortest path from \(e\) to \(B_i\).
	\end{enumerate}
\end{theorem}
\begin{proof}
	The proof is by induction over the number \(k\) of 2-connected components of \(\vegraph{H}\).
	If \(k=0\), none of the conditions is satisfied.
	As there are no 2-connected components, \(H\) is acyclic, so \(H\) is Breaker's win by~\cite[Theorem~3.21]{galliot2025makerbreakersolvedpolynomialtime}.
	If \(k=1\), we distinguish the cases \(e\in B_1\) and \(e\notin B_1\).
	In the first case, only condition~\labelcref{item:thm:block-split-1} can be satisfied.
	Apply \cref{cor:lowdeg-pruning} to \(H\) until this is no longer possible to obtain a hypergraph \(H'\subseteq H\) such that Maker wins on \(H\) if and only if Maker wins on \(H'\).
	We claim that \(H'=H[V_1]\):
	Every edge in \(B_1\) lies on a cycle in \(H\), so they cannot be removed by \cref{cor:lowdeg-pruning}.
	The edges in \(E(H)\setminus B_1\) do not lie on any cycle, so \cref{cor:lowdeg-pruning} can successively remove all of these edges.
	Thus, Maker winning on \(H\) is equivalent to Maker winning on \(H[V_1]\), i.e.~condition~\labelcref{item:thm:block-split-1}.
	
	In the second case, only conditions~\labelcref{item:thm:block-split-2,item:thm:block-split-3} can be satisfied.
	Assume condition~\labelcref{item:thm:block-split-2} is satisfied and let \(v\in V(H)\) be the last vertex on a shortest path from \(e\) to \(B_1\).
	Let \(K\) be the connected component of \(H^{+v}\) containing \(B_1\setminus\{v\}\).
	Set \(H_1\coloneqq H[V(H)\setminus K]\) and \(H_2\coloneqq H[K\cup\{v\}]\).
	The definition of \(H_1\) and \(H_2\) immediately implies \(V(H)=V(H_1)\cup V(H_2)\) and \(V(H_1)\cap V(H_2)=\{v\}\).
	Additionally, every hyperedge \(f\in E(H)\) satisfies \(f\subseteq V(H_1)\) or \(f\subseteq V(H_2)\):
	Assume for contradiction that there exists a hyperedge \(f\in E(H)\) such that \(f\not\subseteq V(H_1)\) and \(f\not\subseteq V(H_2)\).
	Then \(f\) is not the 2-edge \(e\) as \(e\subseteq V(H_1)\).
	Thus, \(f\setminus\{v\}\) is a hyperedge of size at least~2 in \(H^{+v}\) with an endpoint in \(V(H_2)\setminus V(H_1)=K\).
	But this implies \(f\setminus\{v\}\subseteq K\) as \(K\) is a connected component, contradicting the assumption that \(f\not\subseteq V(H_2)\).
	In particular, \(H=H_1\cup H_2\).
	Observe that \(H_2^{+v}\) is connected.
	By \cref{lem:cutpoint-split}, Maker wins on \(H\) if and only if Maker wins on \(H_1\) or on \(H_2^{+v}\).
	As \(v\in B_1\), \(v\) lies on a cycle in \(H[V_1]\subseteq H_2\).
	Thus, \(H_2^{+v}\) is a connected hypergraph containing two 2-edges.
	Such hypergraphs are Maker's win as already observed by Kutz~\cite{kutz2004thesis}.
	
	Now assume condition~\labelcref{item:thm:block-split-2} is not satisfied.
	We will show that Maker's win is equivalent to condition~\labelcref{item:thm:block-split-3}.
	Apply \cref{cor:lowdeg-pruning} to \(H\) until this is no longer possible to obtain a hypergraph \(H'\subseteq H\) such that Maker wins on \(H\) if and only if Maker wins on \(H'\).
	\begin{claim}\label{claim:edge_on_path}
		Every hyperedge-vertex in \(E(H')\) is in \(B_1\) or lies on a path from \(B_1\) to \(e\) in \(\vegraph{H}\).
	\end{claim}
	\begin{claimproof}[Proof of \cref{claim:edge_on_path}]
		The edges in \(B_1\) lie on a cycle and so cannot be removed by \cref{cor:lowdeg-pruning}.
		Similarly, if a hyperedge-vertex lies on a path \(P\) from \(e\) to \(B_1\) in \(\vegraph{H}\), then this path ensures that at most one neighbor of the hyperedge-vertex can have degree~1.
		Thus, such hyperedge-vertex cannot be removed by \cref{cor:lowdeg-pruning} either.
		Assume for contradiction that \(H'\) contains a hyperedge \(f\) that does not satisfy the claim.
		Let \(P=v_0e_1\dots e_\ell v_\ell\) be a longest path in \(H'\) using the edge \(f\).
		Then one of the edges \(e_1\) or \(e_\ell\) is different from \(e\) and does not lie on a cycle.
		This holds because \(f\notin B_1\) and \(f\) does not lie on a path from \(B_1\) to \(e\) in \(\vegraph{H}\).
		Say \(e_1\) has this property.
		In particular, \(v_0\) must have degree~1 in \(\vegraph{H}\).
		But the third vertex in \(e_1\) must also have degree~1 in \(\vegraph{H}\) as otherwise \(P\) could be extended.
		However, then \cref{cor:lowdeg-pruning} could be applied to further reduce \(H'\), a contradiction.
	\end{claimproof}
	
	Let \(v\in V(H)\) be the penultimate vertex on a shortest path from \(e\) to \(B_1\) in \(\vegraph{H}\).
	By \cref{claim:edge_on_path}, \(H'\) consists of \(H[V_1]\) and a path \(P\) from \(v\) ending with the edge \(e\).
	Thus, \(H'=H_1\cup H_2\) where \({H_1=(V(P),E(P))}\) and \(H_2=H[V_1]\).
	Additionally, these hypergraphs satisfy \(V(H_1)\cap V(H_2)=\{v\}\).
	Observe that \(H_2^{+v}\) is connected.
	By \cref{lem:cutpoint-split}, Maker wins on \(H\) if and only if Maker wins on \(H_1\) or on \(H_2^{+v}\).
	However, \(H_1\) is acyclic, so Breaker wins on \(H_1\) by \cite[Theorem~3.21]{galliot2025makerbreakersolvedpolynomialtime}.
	Thus, Maker wins on \(H\) if and only if Maker wins on \(H_2^{+v}\), i.e.~if condition~\labelcref{item:thm:block-split-3} is satisfied.
	This completes the case \(k=1\).
	
	Assume the statement holds for \(k\geq1\).
	We now prove it for \(k+1\).
	Let \(1\leq i<j\leq k\) be such that a shortest path from \(B_i\) to \(B_j\) in \(\vegraph{H}\) does not intersect any of the sets \(B_1,\dots,B_k\) except for \(B_i\) and \(B_j\).
	Without loss of generality, let the indices be such that \(e\) is closer to \(B_i\) than to \(B_j\), i.e.~a shortest path from \(B_i\) to \(e\) in \(\vegraph{H}\) does not go through \(B_j\).
	Let \(v\) be the last element of \(V(H)\) on a shortest path from \(B_i\) to \(B_j\).
	Let \(K\) be the connected component of \(H^{+v}\) containing \(B_j\setminus\{v\}\).
	Then \(H=H_1\cup H_2\) where \({H_1\coloneqq H[V(H)\setminus K]}\) and \({H_2\coloneqq H[K\cup\{v\}]}\).
	Observe that \(V(H_1)\cap V(H_2)=\{v\}\) and that \(H_2^{+v}\) is connected.
	By \cref{lem:cutpoint-split} Maker wins on \(H\) if and only if Maker wins on \(H_1\) or on \(H_2^{+v}\).
	
	The 2-connected component \(B_j\) is not present in \(\vegraph{H_1}\).
	Similarly, \(B_i\) is not present in \(\vegraph{H_2^{+v}}\).
	Thus, \(\vegraph{H_1}\) and \(\vegraph{H_2^{+v}}\) have fewer 2-connected components than \(\vegraph{H}\), so the induction hypothesis can be applied.
	If Maker wins on \(H_1\), \(H_1\) must satisfy one of the conditions~\labelcref{item:thm:block-split-1} through~\labelcref{item:thm:block-split-3}.
	As \(H_1\subseteq H\), \(H\) will satisfy the same condition.
	Thus, we may assume Breaker wins on \(H_1\).
	Then, Maker wins on \(H\) if and only if Maker wins on \(H_2^{+v}\).
	If \(v\in B_j\), then \(H\) satisfies condition~\labelcref{item:thm:block-split-2} as \(v\) is the last vertex on a shortest path from \(e\) to \(B_j\) in \(H\) in this case.
	Additionally, Maker wins on \(H_2^{+v}\) as \(v\) lies on a cycle in \(H[V_j]\subseteq H_2\), so that \(H_2^{+v}\) contains two~2-edges, see \cite{kutz2004thesis}.
	This leaves only the case where \(v\notin B_j\).
	
	By the induction hypothesis, \(H_2^{+v}\) is Maker's win if and only if it satisfies one of the conditions~\labelcref{item:thm:block-split-1} through~\labelcref{item:thm:block-split-3}.
	Let \(f\) be the unique 2-edge in \(H_2^+\)
	Observe that \(H\) cannot satisfy condition~\labelcref{item:thm:block-split-1} because \(H_1\) is Breaker's win and contains the~2-edge \(e\).
	Thus, it remains to show that \(H\) satisfies condition~\labelcref{item:thm:block-split-2} or condition~\labelcref{item:thm:block-split-3} if and only if \(H_2\) satisfies one of the conditions~\labelcref{item:thm:block-split-1} through~\labelcref{item:thm:block-split-3}.
	
	Assume \(H\) satisfies condition~\labelcref{item:thm:block-split-2}.
	Let \(B_\ell\) be a~2-connected component and let \(P\) be the shortest path from \(e\) to \(B_\ell\) in \(H\) witnessing that \(H\) satisfies condition~\labelcref{item:thm:block-split-2}.
	As \(H_1\) does not satisfy any of the conditions, it follows that \(B_\ell\subseteq V(H_2)\).
	Additionally, \(\ell\neq j\) as \(v\notin B_j\).
	Then the subpath of \(P\) starting in the successor of \(v\), which is a 2-edge in \(H_2^{+v}\), witnesses that \(H_2^{+v}\) also satisfies condition~\labelcref{item:thm:block-split-2}.

	Assume \(H\) satisfies condition~\labelcref{item:thm:block-split-3}.
	Let \(B_\ell\) be a~2-connected component and let \(P\) be the shortest path from \(e\) to \(B_\ell\) in \(H\) witnessing that \(H\) satisfies condition~\labelcref{item:thm:block-split-3}.
	As \(H_1\) does not satisfy any of the conditions, it follows that \(B_\ell\subseteq V(H_2)\).
	If \(\ell=j\), then Maker wins on \(H[V_j]^{+v}\subseteq H_2^{+v}\), so \(H_2^{+v}\) satisfies one of the conditions by the induction hypothesis.
	Otherwise, \(\ell\neq j\).
	Then the subpath of \(P\) starting in the successor of \(v\), which is a 2-edge in \(H_2^{+v}\), witnesses that \(H_2^{+v}\) also satisfies condition~\labelcref{item:thm:block-split-3}.
	
	Assume \(H_2^{+v}\) satisfies condition~\labelcref{item:thm:block-split-1}.
	Let \(R\) be the 2-connected component satisfying condition~\labelcref{item:thm:block-split-1}.
	Then \(f\in R\), so \(R\) was created from \(V_j\) by removing \(v\) from the edges of \(H_2\).
	Thus, \(H\) satisfies condition~\labelcref{item:thm:block-split-3} as witnessed by \(B_j\).
	
	Finally, assume \(H_2^{+v}\) satisfies condition~\labelcref{item:thm:block-split-2} or condition~\labelcref{item:thm:block-split-3}.
	Then \(H\) satisfies the same condition because any 2-connected component \(R\) of \(\vegraph{H_2^{+v}}\) that does not use the new 2-edge \(f\) in \(H_2^{+v}\) is also a 2-connected component of \(\vegraph{H}\) and the last two vertices of a shortest path from \(f\) to \(R\) in \(\vegraph{H_2^{+v}}\) are the same as the last two vertices of a shortest path from \(e\) to \(R\) in \(\vegraph{H}\) due to \(e\in E(H_1)\).
\end{proof}

Combining \cref{thm:kutz-analysis,thm:block-split}, we obtain the algorithm for connected hypergraphs that are not 3-uniform.
This is used as a subroutine in the algorithm for arbitrary hypergraphs with rank at most~3.

\begin{corollary}\label{cor:kutz-non-3-uniform}
	Let \(H\) be a connected linear hypergraph of rank at most~3 that is not 3-uniform.
	Then the winner of the Maker-Breaker game played on \(H\) can be determined in time \(\Oh{|V(H)|+|E(H)|}\).
\end{corollary}
\begin{proof}
	If \(H\) contains a hyperedge of size one, Maker trivially wins.
	Similarly, if \(H\) contains at least two~2-edges, it is also Maker's win as observed by Kutz~\cite{kutz2004thesis}.
	Thus, we may assume that \(H\) is an almost~3-uniform hypergraph.
	Let \(e\) be the unique~2-edge in \(H\).
	Our algorithm proceeds as follows:
	Compute \(\vegraph{H}\) from \(H\) in time \(\Oh{|V(H)|+|E(H)|}\) using \cref{thm:kutz-analysis}.
	Then compute the maximal 2-connected components \(B_1,\dots,B_k\) of \(\vegraph{H}\) with size at least~3 in time \(\Oh{|V(\vegraph{H})|}\) using a standard linear time algorithm, e.g.~\cite{DBLP:journals/cacm/HopcroftT73}.
	If no such components exist, \(\vegraph{H}\) is acyclic, so \(H\) is Breaker's win, see \cite[Theorem~3.21]{galliot2025makerbreakersolvedpolynomialtime}.
	
	Otherwise, \(k\geq1\).
	Apply \cref{thm:block-split} to \(H\).
	If any \(B_i\) satisfies condition~\labelcref{item:thm:block-split-2}, we immediately find that \(H\) is Maker's win.
	Otherwise, we need to verify whether condition~\labelcref{item:thm:block-split-1} or condition~\labelcref{item:thm:block-split-3} is satisfied.
	For \(1\leq i\leq k\), set \(H_i=H[V(H)\cap N[B_i]]\) if \(e\in E(H_i)\).
	For \(1\leq i\leq k\), set \(H_i=H[V(H)\cap N[B_i]]^{+v_i}\) if \(e\notin E(H_i)\) where \(v_i\) is the last element of \(V(H)\) on a shortest path from \(e\) to \(B_i\) in \(\vegraph{H}\).
	These hypergraphs correspond to condition~\labelcref{item:thm:block-split-1} and~\labelcref{item:thm:block-split-3} of \cref{thm:block-split}, respectively.
	By \cref{thm:block-split}, \(H\) is Maker's win if and only one of \(H_1,\dots,H_k\) is Maker's win.
	
	For each \(1\leq i\leq k\), \(H_i\) is a connected almost 3-uniform hypergraph such that \(\vegraph{H_i}\) contains only a single~2-connected component, namely \(B_i\), and \(V(\vegraph{H_i})=N[B_i]\).
	Thus, Breaker's win on any of the hypergraphs \(H_1,\dots,H_k\) is characterized by \cref{thm:kutz-breaker-chara}.
	Then the winner of the Maker-Breaker games on \(H_1,\dots,H_k\) can be computed in time \(\Oh{\sum_{i=1}^k\left(|V(H_i)|+|E(H_i)|\right)}\) using \cref{thm:kutz-analysis}.
	This determines the winner of the game on \(H\) in the stated running time.
\end{proof}

It seems unlikely that the running time in \cref{cor:kutz-non-3-uniform} can be improved.
Any strictly faster algorithm does not even have enough time to read the entirety of the input for sufficiently large instances.

In order to apply \cref{thm:kutz-breaker-chara} to a hypergraph \(H\), \(H\) must contain a~2-edge.
This is achieved by simulating the first round of the game through brute-force enumeration.
\Cref{lem:2-edge-creation} is central for the correctness of this approach.

\begin{lemmaq}[{\cite[Lemma~24]{kutz2004thesis}}]\label{lem:2-edge-creation}
	Let \(H\) be a 3-uniform hypergraph that is Maker's win.
	Then there exists a vertex \(v\in V(H)\) such that for every vertex \(w\in V(H)\), the hypergraph \(H^{+v}-w\) has a connected component \(K\) such that \(H[K]\) contains a~2-edge and is Maker's win.
\end{lemmaq}

We can now state the main result of this section, generalizing \cref{cor:kutz-non-3-uniform} to linear hypergraphs of rank at most~3.

\begin{theorem}\label{thm:general-alg}
	There exists an algorithm that, given a linear hypergraph \(H\) of rank at most~3, determines the winner of the Maker-Breaker game played on \(H\) in time \(\Oh{|V(H)|^3+|V(H)|^2|E(H)|}\).
\end{theorem}
\begin{proof}
	Without loss of generality, we may assume that \(H\) is connected as the algorithm can be applied to every connected component of \(H\) while retaining the running time bound.
	If \(H\) is not 3-uniform, the algorithm of \cref{cor:kutz-non-3-uniform} immediately yields the result, so assume that \(H\) is~3-uniform.
	Observe that, by definition of the Maker-Breaker game, Breaker wins on \(H\) if and only if
	\begin{equation*}
		\forall v\in V(H)\,\exists v\neq w\in V(H):\text{ Breaker wins }H^{+v}-w.
	\end{equation*}
	To check if this property is satisfied, we use the following algorithm:
	Enumerate all pairs \((v,w)\) in \({V(H)\times V(H)}\) with \(v\neq w\) and determine the outcome of the game on \(H'=H^{+v}-w\).
	By \cref{lem:2-edge-creation}, it is enough to confirm that every connected components of \(H'\) containing a~2-edge is Breaker's win.
	Let \(K_1,\dots,K_\ell\) be the connected component of \(H'\) containing a~2-edge.
	Then \(H[K_i]\) is not 3-uniform.
	Applying \cref{cor:kutz-non-3-uniform} to \(H'[K_i]\) for every \(1\leq i\leq\ell\) determines if Breaker wins on every \(H_i\), \(1\leq i\leq\ell\), in time \(\Oh{\sum_{i=1}^{\ell}\left(|V(H_i)|+|E(H_i)|\right)}=\Oh{|V(H)|+|E(H)|}\).
	We consider at most \(|V(H)|^2\) pairs \((v,w)\) and processing each pair takes time \(\Oh{|V(H)|+|E(H)|}\) resulting in the stated running time bound.
\end{proof}

If the hypergraph is connected and not 3-uniform, \cref{cor:kutz-non-3-uniform} is preferable to \cref{thm:general-alg} as it guarantees a better running time.
In \cref{sec:k4-chara}, we use this fact to improve the running time in the special case of the unbiased triangle game for graphs containing \(K_4\) as a subgraph.

To obtain an algorithm for the unbiased triangle game from \cref{thm:general-alg}, the only missing component is the efficient computation of \(\gamehg{G}\) from a graph \(G\).
We will later need to compute \(\etgraph{G}\) efficiently, so we consider this task here as well.

\begin{lemma}\label{lem:computing-etg}
	Given a graph \(G\), both \(\gamehg{G}\) and \(\etgraph{G}\) can be computed in time \(\Oh{n+m^{1.5}}\).
\end{lemma}
\begin{proof}
	The following algorithm achieves the desired running time:
	First, remove all isolated vertices from \(G\) in time \(\Oh{n+m}\).
	Then, compute \(\trias{G}\) in time \(\Oh{m^{1.5}}\) using \cite[Proposition~1]{eisenbrand2004fixedparameterclique}.
	Compute \(\gamehg{G}\) by starting with the graph \(F=\{E(G),\emptyset\}\).
	For each \(t\in\trias{G}\), add the edge \(E(G[t])\) to \(F\).
	After this process, \(F\) is equal to \(\gamehg{G}\).
	
	To compute \(\etgraph{G}\), start with the graph \(F'=\{E(G)\cup\trias{G},\emptyset\}\).
	For each \(t\in\trias{G}\) and for each \(e\in t\), add the edge \(\{e,t\}\) to \(F'\).
	After the loop is done, \(F'\) is equal to \(\etgraph{G}\).
	
	The algorithm is clearly correct.
	Note that \(|\trias{G}|=\Oh{m^{1.5}}\) as otherwise it would not be possible to compute the set in this time.
	Thus, the algorithm runs in time~\(\Oh{n+m^{1.5}}\).
\end{proof}

\begin{corollary}
	The winner of the unbiased triangle game on a given board graph \(G\) can be determined in time \({\Oh{\generalTime}}\).
\end{corollary}
\begin{proof}
	Compute \(\gamehg{G}\) from \(G\) using \cref{lem:computing-etg}.
	Then apply \cref{thm:general-alg} to \(\gamehg{G}\) to determine the winner of the Maker-Breaker game on \(\gamehg{G}\).
	As \(|V(\gamehg{G})|=m\) and \(|E(\gamehg{G})|=t\), the running time is bounded by
	\begin{equation*}
		\Oh{n+m^{1.5}+|V(\gamehg{G})|^3+|V(\gamehg{G})|^2|E(\gamehg{G})|}=\Oh{\generalTime}.\qedhere
	\end{equation*}
\end{proof}

\subsection[Classification of Graphs containing \texorpdfstring{\(K_4\)}{the complete Graph on four Vertices}]{Classification of Graphs containing \texorpdfstring{\boldmath\(K_4\)}{the complete Graph on four Vertices}}\label{sec:k4-chara}
When the board graph \(G\) contains \(K_4\) as a subgraph and its edge-triangle incidence graph is connected, the outcome of the game on \(\gamehg{G}\) can be determined using \cref{thm:k4-elimination}.
We require the following results:

\begin{lemmaq}[{\cite[Lemma 19]{kutz2004thesis}}]\label{lem:kutz-quasi-articulation}
	Let \(H=A\cup B\) be the union of two hypergraphs \(A\) and \(B\) whose vertex sets intersect in a single element, i.e.~\({V(A)\cap V(B)=\{v\}}\).
	Then Maker wins on \(H\) if and only if one of the following conditions is satisfied:\\
	\begin{minipage}{.3\linewidth}
		\begin{enumerate}[label=(\arabic*)]
			\item Maker wins on \(A\).
		\end{enumerate}
	\end{minipage}
	\hfill
	\begin{minipage}{.3\linewidth}
		\begin{enumerate}[label=(\arabic*)]
			\setcounter{enumi}{1}
			\item Maker wins on \(B\).
		\end{enumerate}
	\end{minipage}
	\begin{minipage}{.39\linewidth}
		\begin{enumerate}[label=(\arabic*)]
			\setcounter{enumi}{2}
			\item Maker wins on \(A^{+v}\) and on \(B^{+v}\).\qedhere
		\end{enumerate}
	\end{minipage}
\end{lemmaq}

\begin{corollaryq}\label{cor:k4-preprocessing}
	Let \(H=A\cup B\) be the union of two hypergraphs \(A\) and \(B\) with \(V(A)\cap V(B)=\{v\}\) such that Breaker wins on \(B^{+v}\).
	Then Maker wins on \(H\) if and only if Maker wins on \(A\).
\end{corollaryq}

\begin{figure}[t]
	\centering
	\begin{equation*}
		\begin{array}{c@{\hspace{10mm}}c}
			\vcenter{\hbox{\includegraphics[scale=0.8]{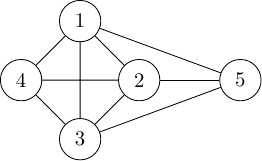}}} &
			\vcenter{\hbox{\includegraphics[scale=0.8]{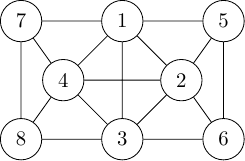}}} \\[2pt]
			H_1 & H_2 \\[4pt]
			\vcenter{\hbox{\includegraphics[scale=0.6]{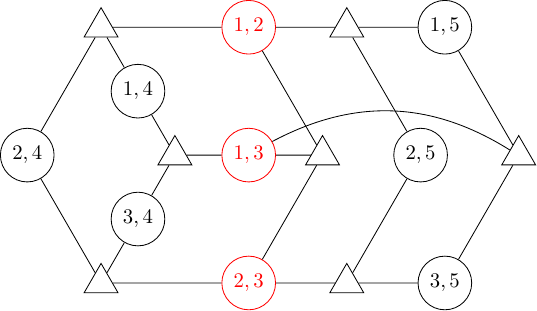}}} &
			\vcenter{\hbox{\includegraphics[scale=0.6]{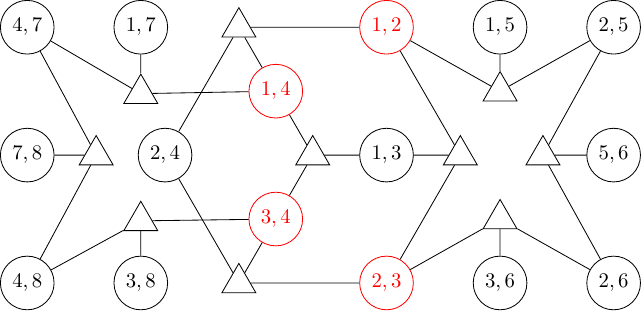}}} \\[2pt]
			\etgraph{H_1} & \etgraph{H_2}
		\end{array}
	\end{equation*}
	\caption{The graphs \(H_1\) and \(H_2\) as well as their edge-triangle incidence graphs.
		The graph \(H_1\) satisfies condition~\labelcref{item:k4-chara-3-in-comp} of \cref{thm:k4-elimination}.
		The three vertices in the same connected component of \(\etgraph{H_1}-\trias{\{1,2,3,4\}}\) are colored red.
		The graph \(H_2\) satisfies condition~\labelcref{item:k4-chara-good-anchor} of \cref{thm:k4-elimination}.
		Two pairs of vertices in the same connected components of \(\etgraph{H_2}-\trias{\{1,2,3,4\}}\) witnessing this are colored red.}
	\label{fig:k4-elimination}
\end{figure}

\begin{theorem}\label{thm:k4-elimination}
	Let \(G=(V,E)\) be such that \(\etgraph{G}\) is connected and \(K_4\cong G[U]\) for some \(U\subseteq V\).
	Set \({E^*=E(G[U])}\).
	For \(e\in E^*\), let \(K_e\) be the connected component of \(e\) in \(\etgraph{G}-N[E^*\setminus\{e\}]\).
	Then Maker wins on \(\gamehg{G}\) if and only if one of the following conditions is satisfied:
	\begin{enumerate}[label=(\arabic*)]
		\item\label{item:k4-chara-3-in-comp} Three elements of \(E^*\) are in the same connected component of \(\etgraph{G}-\trias{U}\).
		\item\label{item:k4-chara-biased-win} There exists an edge \(e\in E^*\) such that Maker wins \(\gamehg{G[E\cap K_e]}^{+e}\).
		\item\label{item:k4-chara-good-anchor} There exists a connected component \(K\) of \(\etgraph{G}-\trias{U}\) containing exactly two edges \(e_1,e_2\) of \(E^*\) such that Maker wins on \(\gamehg{G}+\{e_1,e_2\}\).
	\end{enumerate}
\end{theorem}
\begin{proof}
	Every edge \(e\in E\) is contained in some \(K_{e'}\) for an edge \(e'\in E^*\):
	Let \(e\in E\) and \(e'\in E ^*\) be arbitrary.
	If a shortest path from \(e\) to \(e'\) in \(\etgraph{G}\) does not intersect \(\trias{U}\), then \(e\) lies in \(K_{e'}\).
	Otherwise, a shortest path from \(e\) to \(e'\) contains a triangle \({t\in\trias{U}}\).
	As \(N(\trias{U})\subseteq E^*\), \(e\) must be able to reach some element \(e''\) of \(E^*\) without using a triangle from \(\trias{U}\), so \(e\) is in \(K_{e''}\).
	
	First, we will show that the conditions are sufficient for Maker's win.
	For the analysis, we consider a special case of condition~\labelcref{item:k4-chara-good-anchor} as a separate condition:
	\begin{enumerate}[label=(\arabic*)]\setcounter{enumi}{3}
		\item\label{item:k4-chara-double-cycle} There are two connected components of \(\etgraph{G}-\trias{U}\) containing exactly 2 edges of \(E^*\).
	\end{enumerate}
	
	The proof commonly uses the fact that if Maker claims an edge from \(E^*\) when neither Maker nor Breaker claimed an edge from \(E^*\) before, Breaker must also claim an edge from \(E ^*\) to avoid losing.
	This is due to the fact that every edge in \(E^*\) lies on a cycle with every other edge in \(E^*\) in \(\etgraph{G[U]}\), see also \cref{fig:etgraph-example,fig:k4-elimination}.
	Thus, Maker can claim two edges in the same connected component of \(\etgraph{G[U]}\) in the next round by claiming any available edge in \(E^*\) if Breaker does not answer in \(E ^*\).
	
	\labelcref{item:k4-chara-3-in-comp}: Let \(e_0,e_1,e_2\in E^*\) be pairwise distinct and in the same connected component \(K\) of \(\etgraph{G}-\trias{U}\) such that the shortest paths \(P_1,P_2\) from \(e_0\) to \(e_1\) and from \(e_0\) to \(e_2\) minimize \(|V(P_1)|+|V(P_2)|\) among all possible choices of such vertices.
	Note that \(N(\triasCap{P_i})\cap E^*=\{e_0,e_i\}\) for all \(i\in\{1,2\}\), i.e.~the only edges from \(E^*\) used by triangles in \(P_i\) are \(e_0\) and \(e_i\).
	No triangle-vertex in \(P_i\) except for the last one can have a neighbor in \(E^*\setminus\{e_0\}\) as otherwise \(|V(P_1)|+|V(P_2)|\) is not minimal, a contraindication.
	The last triangle-vertex in \(P_i\) cannot have a neighbor in \(E^*\) besides \(e_i\) as then it would use two edges from \(E^*\) and would thus be in \(\trias{U}\), which is disjoint from \(K\).
	This situation is visualized in \cref{fig:k4-elimination} by the graph \(H_1\) where \(\left(e_0,e_1,e_2\right)=\left(\{1,3\},\{1,2\},\{2,3\}\right)\).
	
	Maker first picks \(e_0\).
	Breaker must respond by picking some edge \({a_1\in E^*}\) as discussed at the beginning of the proof.
	Regardless of the choice of \(a_1\), \({\etgraph{G[U]}-N[a_1]}\) is a tree and therefore connected, see \cref{fig:etgraph-example,fig:k4-elimination}.
	At least one of \(e_1\) and \(e_2\), say \(e_1\), is not claimed by Breaker and is thus reachable from \(e_0\) in \(\etgraph{G[U]}-N[a_1]\) by a path \(Q\).
	Breaker cannot separate \(e_0\) and \(e_1\) in \({\etgraph{G}-N[a_1]}\) as witnessed by the \((e_0,e_1)\)-paths \(Q\) and \(P_1\).
	The path \(Q\) can only be destroyed by picking edges from \(E^*\) and the path \(P_1\) can only be destroyed by picking edges from \(N(\triasCap{P_1})\).
	However, \(E^*\cap N(\triasCap{P_1})=\{e_0,e_1\}\) and both of these edges have already been claimed by Maker, so at least one of the paths remains after Breaker's move.
	Thus, Maker wins by \cref{lem:maker-winning-structure}.
	
	\labelcref{item:k4-chara-biased-win}: Let \(e\in E^*\) be such that Maker wins \(\gamehg{G[E\cap K_e]}^{+e}\).
	Maker wins on \(\gamehg{G}\) by first picking \(e\).
	Breaker must reply with some \(a\in E^*\) as discussed at the beginning of the proof.
	Then Maker wins by following a winning strategy for the game \(\gamehg{G[E\cap K_e]}^{+e}\).
	This is possible because \(a\in E^*\setminus\{e\}\), so \(N[a]\cap K_e=\emptyset\), i.e.~Breaker's move has no impact on the game on \(\gamehg{G[E\cap K_e]}^{+e}\).
	
	\labelcref{item:k4-chara-double-cycle}: Let \(e_1,e_2\in E ^*\) and \(e_3,e_4\in E^*\) be in the same connected components \(K_1,K_2\) of \(\etgraph{G}-\trias{U}\), respectively.
	A triangle in \(G\) is uniquely determined by two of its component edges, so there can be at most one triangle in \(G\) using~3 of the edges \(e_1,\dots,e_4\).
	Thus, we may assume that the labeling of the edges is such that if there exists a triangle in \(G\) using~3 of the edges \(e_1,\dots,e_4\), then this triangle consists of \(e_1,e_2,e_4\).
	All edges in \(E^*\) except one lie on a common triangle with \(e_3\) as \(G[U]\cong K_4\), see~\cref{fig:etgraph-example}.
	Therefore, we may assume that \(e_3\) lies on a common triangle with \(e_1\) (swap \(e_1\) and \(e_2\) otherwise, this causes no problem with the previous assumption).

	We consider two cases.
	In the first case, \(e_2\) and \(e_3\) do not lie on a common triangle.
	Maker wins as follows:
	First, Maker picks \(e_1\).
	As discussed at the beginning of the proof, Breaker must reply with some edge \(a_1\in E^*\).
	The graph \(\etgraph{G[U]}-N[a_1]\) is a tree and therefore connected, so there exists an \((e_1,e_2)\)-path \(Q\) in \(\etgraph{G[U]}-N[a_1]\) if \(a_1\neq e_2\).
	Thus, \(e_1\) lies on a cycle with \(e_2\) using the path between the two in \(K_1\) and the path \(Q\).
	This leads to Maker's win in the next round by \cref{lem:maker-winning-structure}.
	Therefore, Breaker must respond by picking \(a_1=e_2\).
	Next, Maker picks \(e_3\) which results in a threat in \(G[U]\).
	This threat is the unique triangle \(t_1\) using \(e_1\) and \(e_3\) in \(\trias{U}\):
	By the assumption on the labeling, there is at most one triangle using three edges from \(\{e_1,e_2,e_3,e_4\}\) and this triangle uses the edges \(e_1,e_2,e_4\).
	Thus, the third edge \(f\in E^*\) in the triangle \(t_1\) is neither \(e_2\) nor \(e_4\), so it is unaffected by Breaker's choice of \(a_1=e_2\).
	Breaker must respond with \(a_2=f\neq e_4\).
	Finally, Maker picks \(e_4\).
	Because \(e_2\) and \(e_3\) do not lie on a common triangle in \(G[U]\cong K_4\), \(e_3\) and \(e_4\) must lie on a common triangle \(t_2\) with third edge \(g\in E^*\).
	This triangle is a relevant threat:
	The edge \(a_1=e_2\) does not lie on a common triangle with \(e_3\), so \(a_1\neq g\).
	Additionally, \(f=g\) is impossible as this would imply \(t_1=t_2\), so \(e_1,e_2\) and \(e_3\) would form a triangle contradicting the assumption on the labeling of the edges.
	Thus, Breaker plays \(a_3=g\) to avoid losing in the next round.
	But \(e_3\) and \(e_4\) are connected by a path in \(K_2\), so Maker has claimed two edges in the same connected component of \(\etgraph{G}-N[\{a_1,a_2,a_3\}]\) leading to Maker's win by \cref{lem:maker-winning-structure}.
	This situation is visualized in \cref{fig:k4-elimination} by the graph \(H_2\).
	The edges \(e_1,\dots,e_4\) can be chosen as \(\left(e_1,e_2,e_3,e_4\right)=\left(\{1,2\},\{2,3\},\{1,4\},\{3,4\}\right)\).
	
	In the second case, \(e_2\) and \(e_3\) lie on a common triangle \(t_1\) with third edge \(f\in E^*\).
	Then Maker wins as follows:
	First, Maker picks \(e_3\).
	As discussed at the beginning of the proof, Breaker must reply with some edge \(a_1\in E^*\).
	The graph \(\etgraph{G[U]}-N[a_1]\) is a tree, so there exists an \((e_3,e_4)\)-path \(Q\) in \(\etgraph{G[U]}-N[a_1]\) if \(a_1\neq e_4\).
	Thus, \(e_3\) lies on a cycle with \(e_4\) using the path between the two in \(K_2\) and the path \(Q\).
	This leads to Maker's win in the next round by \cref{lem:maker-winning-structure}.
	Therefore, Breaker must respond by picking \(a_1=e_4\).
	Next, Maker picks \(e_1\) which results in a threat in \(G[U]\).
	This threat is the unique triangle \(t_2\) using \(e_1\) and \(e_3\) in \(\trias{U}\):
	By the assumption on the labeling, there is at most one triangle using three edges from \(\{e_1,e_2,e_3,e_4\}\) and this triangle uses the edges \(e_1,e_2,e_4\).
	Thus, the third edge \(g\in E^*\) in the triangle \(t_2\) is neither \(e_2\) nor \(e_4\), so it is unaffected by Breaker's choice of \(a_1=e_4\).
	Finally, Maker picks \(e_2\).
	The triangle \(t_1\) is now a threat:
	By the assumption on the labeling, any triangle using three edges from \(\{e_1,e_2,e_3,e_4\}\) uses the edges \(e_1,e_2,e_4\), so \(f\neq a_1=e_4\).
	Additionally, \(f\neq g\) as otherwise \(t_1=t_2\) implying that \(e_1,e_2\) and \(e_3\) form a triangle, a contradiction.
	Thus, Breaker responds with \(a_3=f\) to avoid losing in the next round.
	But \(e_1\) and \(e_2\) are connected by a path in \(K_1\), so Maker claimed two edges in the same connected component of \(\etgraph{G}-N[\{a_1,a_2,a_3\}]\) leading to Maker's win by \cref{lem:maker-winning-structure}.
	
	\labelcref{item:k4-chara-good-anchor}: We may assume that none of the other conditions is satisfied.
	In particular, any connected component of \(\etgraph{G}-\trias{U}\) except \(K\) contains only a single element of \(E^*\) as conditions \labelcref{item:k4-chara-3-in-comp,item:k4-chara-double-cycle} are not satisfied and Breaker wins \(\gamehg{G[E\cap K_e]}^{+e}\) for all \(e\in E^*\) as \labelcref{item:k4-chara-biased-win} is not satisfied.
	For any \({e\in E^*\setminus\{e_1,e_2\}}\), \(\gamehg{G}\) decomposes into \(H_1=\gamehg{G}-\left(K_e\setminus\{e\}\right)\) and \(H_2=\gamehg{G[E\cap K_e]}\) with \(V(H_1)\cap V(H_2)=\{e\}\).
	By \cref{cor:k4-preprocessing}, we may therefore assume that \(K_e=\{e\}\) for all \(e\in E^*\setminus\{e_1,e_2\}\).
	Assume for contradiction that Breaker wins on \(\gamehg{G}\).
	We will show that Breaker then also wins on \(\gamehg{G}+\{e_1,e_2\}\).
	
	\begin{claim}\label{claim:k4-beaker-choice}
		There exists a winning strategy for Breaker on \(\gamehg{G}\) such that its first claimed edge from \(E^*\) against any Maker strategy is an element of \(\{e_1,e_2\}\).
	\end{claim}
	\begin{claimproof}[Proof of \cref{claim:k4-beaker-choice}]
		Let any Maker strategy \(M\) be given and consider the first round \(i+1\) in which a fixed winning strategy \(B\) for Breaker claims an edge \(a\) from \(E^*\) playing against \(M\).
		If \(a\in\{e_1,e_2\}\), there is nothing to show, so assume this is not the case.
		We have already shown that Breaker must claim an edge from \(E^*\) when Maker does so for the first time.
		Therefore, no element of \(E^*\) was claimed by Maker or Breaker before round \(i+1\).
		Let \(E_b^i\) be the set of edges Breaker claimed up to and including round \(i\).
		Let \(K\) be the connected component of \(\etgraph{G}-N[E_b^i]\) containing \(E^*\).
		The graph \(\etgraph{G[U]}-N[e]\) is a tree for any \(e\in E^*\).
		Additionally, \(\{e_1,e_2\}\) separates \(\etgraph{G[U]}\) from the rest of \(K\), so \(K-N[a]\) is connected.
		Thus, at most one edge in \(K\) can be claimed by Maker after round \(i+1\) by \cref{lem:maker-winning-structure}.
		In particular, one of \(e_1\) and \(e_2\) is not claimed by Maker in round \(i+1\).
		Let \(a^*\in\{e_1,e_2\}\) be such an edge.
		Breaker wins on all connected components of \(\etgraph{G}-N[E_b^i]\) by playing according to \(B\).
		Thus, it is enough to show Breaker wins on \(K\) even if Breaker claims \(a^*\) instead of \(a\) in round \(i+1\).

		If Maker has no edge claimed in \(K-N[a^*]\) after round \(i+1\), Breaker can clearly win on \({G[E\cap K]}\) as Breaker wins on \(G\).
		Otherwise, Maker has exactly one edge \(e_m\) claimed in \(K-N[a^*]\).
		Let \({a'\in\{e_1,e_2\}\setminus\{a^*\}}\) and set \({H=\gamehg{G[K\setminus\{a^*\}]}^{+e_m}}\).
		Assume that \(e_m\neq a'\).
		Then \(H\) decomposes into \({H=H_1\cup H_2}\) with \({H_1=H[E^*\setminus\{a^*\}]}\) and \(H_2=H[E\cap K\cap K_{a'}]\) such that \({V(H_1)\cap V(H_2)=\{a'\}}\).
		We distinguish two cases:
		
		First, assume that \(e_m\in E^*\).
		Breaker wins on \(H_2^{+a'}\) as Breaker wins on \(\gamehg{G[E\cap K]}^{+a'}\) by playing according to \(B\).
		Applying \cref{cor:k4-preprocessing} reduces the outcome of the game on \(H\) to the outcome of the game on \(H_1\).
		Breaker wins on \(H_1\) as \(H_1\) is acyclic.
		
		Now assume \(e_m\in K\setminus E^*\).
		Breaker wins on \(H_1^{+a'}\) as \(H_1\) is acyclic.
		Applying \cref{cor:k4-preprocessing} reduces the outcome of the game on \(H\) to the outcome of the game on \(H_2\).
		Breaker wins on \(\gamehg{G[E\cap K_{a'}]}^{+e_m}\) by playing according to \(B\).
		As \(H_2\subseteq\gamehg{G[E\cap K_{a'}]}^{+e_m}\), Breaker also wins on \(H_2\).
		
		Finally, assume \(e_m=a'\).
		Then \(H_1\coloneqq H[N_H[E^*\setminus\{a^*\}]]\) and \(H_2\coloneqq H-V(H_1)\) are not connected in \(H\).
		Breaker wins on \(H_2\) as \(H_2\subseteq\gamehg{G[E\cap K_{e_m}]}^{+e_m}\) and Breaker wins \(\gamehg{G[E\cap K_{e_m}]}^{+e_m}\).
		As \(H_1\) is a forest, Breaker wins on \(H_1\), so that Breaker also wins on \(H\).
	\end{claimproof}
	
	Let a winning strategy \(M\) for Maker in \(\gamehg{G}+\{e_1,e_2\}\) be given.
	As \(\gamehg{G}+\{e_1,e_2\}\) and \(\gamehg{G}\) have the same vertices, we can also view the Breaker strategy \(B\) of \cref{claim:k4-beaker-choice} for the Maker-Breaker game on \(\gamehg{G}\) as a strategy for the Maker-Breaker game on \(\gamehg{G}+\{e_1,e_2\}\).
	Consider the game of \(M\) against \(B\) on \(\gamehg{G}+\{e_1,e_2\}\).
	As \(G[U]\) contains triangles, the strategy \(B\) must claim at least one element of \(E^*\), which by \cref{claim:k4-beaker-choice} is an element of \(\{e_1,e_2\}\).
	Thus, \(M\) does not claim both elements \(e_1\) and \(e_2\) of the~2-edge \(\{e_1,e_2\}\).
	This is the only hyperedge in \(\gamehg{G}+\{e_1,e_2\}\) that does not exist in \(\gamehg{G}\), so the winning strategy \(M\) claims all elements of a hyperedge from \(\gamehg{G}\).
	Then \(M\) also wins against \(B\) on \(\gamehg{G}\).
	However, this is a contradiction as \(B\) is a winning strategy for Breaker on \(\gamehg{G}\).
	
	We now show that one of the conditions~\labelcref{item:k4-chara-3-in-comp} to~\labelcref{item:k4-chara-good-anchor} is necessary for Maker's win.
	To that end, we give a winning strategy for Breaker assuming that none of the conditions are satisfied.
	As~\labelcref{item:k4-chara-3-in-comp} is not satisfied, every connected component of \(\etgraph{G}-\trias{U}\) contains at most two elements of \(E^*\).
	If there exists a connected component of \(\etgraph{G}-\trias{U}\) containing exactly two elements \(e_1\) and \(e_2\) of \(E^*\), Breaker wins on \(\gamehg{G}+\{e_1,e_2\}\) as~\labelcref{item:k4-chara-good-anchor} is not satisfied.
	In particular, Breaker also wins on \(\gamehg{G}\).
	Otherwise, every connected component of \(\etgraph{G}-\trias{U}\) contains exactly one element of \(E^*\), so \(\gamehg{G}\) decomposes into \(H_1=\gamehg{G}-\left(K_e\setminus\{e\}\right)\) and \(H_2=\gamehg{G[E\cap K_e]}\) with \(V(H_1)\cap V(H_2)=\{e\}\) for every \(e\in E^*\).
	Breaker wins \(\gamehg{G[E\cap K_e]}^{+e}\) for all \(e\in E^*\) as~\labelcref{item:k4-chara-biased-win} is not satisfied.
	Thus, we may assume that \(G\cong K_4\) by \cref{cor:k4-preprocessing}.
	
	It remains to show that Breaker wins on \(\gamehg{K_4}\).
	Let \(V(K_4)=\{1,2,3,4\}\) as in \cref{fig:etgraph-example}.
	By symmetry, we may assume that Maker first claims the edge \(e_1=\{1,2\}\).
	Breaker responds with the edge \({a_1=\{2,3\}}\).
	The state of the game is visualized in \cref{fig:k4-loss}~(a).
	If Maker does not claim \(e_2=\{1,4\}\), Breaker claims \({a_2=\{1,4\}}\).
	Then Breaker has already won as he has claimed one edge from each triangle.
	This is visualized in \cref{fig:k4-loss}~(b).
	Otherwise, Maker claims \(e_2=\{1,4\}\) and Breaker is forced to reply with \(a_2=\{2,4\}\).
	This situation is visualized in \cref{fig:k4-loss}~(c).
	Finally, Breaker claims the last remaining edge after Maker's move and wins.
\end{proof}

The conditions~\labelcref{item:k4-chara-3-in-comp,item:k4-chara-good-anchor} can also be summarized as follows:
Whenever two different edges \(e_1\) and \(e_2\) from \(E^*\) lie in the same connected component of \(\etgraph{G}-\trias{U}\), the~2-edge \(\{e_1,e_2\}\) may be added to \(\gamehg{G}\) without changing the outcome of the game.
As a connected hypergraph containing at least two size~2 hyperedges is always Maker's win, this captures the conditions~\labelcref{item:k4-chara-3-in-comp,item:k4-chara-good-anchor}.

Using \cref{thm:k4-elimination}, we obtain the algorithm for board graphs containing \(K_4\) as a subgraph.

\begin{figure}[t]
	\centering
	\[
	\begin{array}{c@{\hspace{1em}}c@{\hspace{1em}}c}
		\vcenter{\hbox{\includegraphics[scale=0.695]{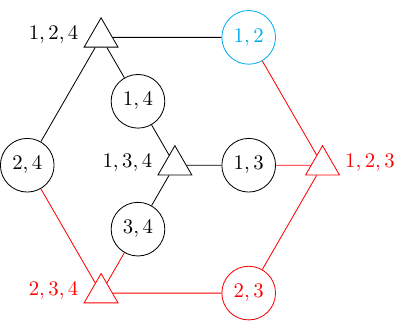}}} &
		\vcenter{\hbox{\includegraphics[scale=0.695]{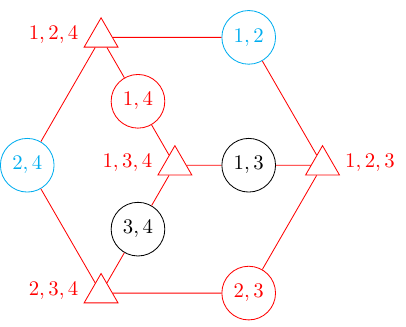}}} &
		\vcenter{\hbox{\includegraphics[scale=0.695]{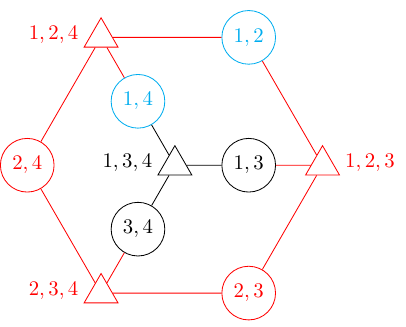}}} \\[3pt]
		\text{(a) The state after the first round.} &
		\text{(b) Maker does not claim \(\{1,4\}\)} &
		\text{(c) Maker claims \(\{1,4\}\)}\\
		&
		\text{in the second round.} &
		\text{in the second round.} 
	\end{array}
	\]
	\caption{The edge-triangle incidence graph \(\etgraph{K_4}\) during the unbiased triangle game.
		The edge-vertices are named as in \cref{fig:etgraph-example}.
		Cyan edge-vertices are claimed by Maker, red edge-vertices are claimed by Breaker.
		Triangle-vertices and edges made irrelevant for Maker by Breaker's choices are also colored red.}
	\label{fig:k4-loss}
\end{figure}

\begin{theorem}\label{thm:k4-alg}
	Let \(G\) be a graph such that \(\etgraph{G}\) is connected and \(K_4\subseteq G\).
	Then the winner of the game on \(\gamehg{G}\) can be determined in time \(\Oh{\cliqueTime}\).
\end{theorem}
\begin{proof}
	Compute \(\etgraph{G}\) using \cref{lem:computing-etg}.
	To obtain a set \(U\subseteq V(G)\) such that \(G[U]\cong K_4\), we run two algorithms in parallel and choose \(U\) from the output of whichever one terminates first.
	The first algorithm is the one of Dalirrooyfard et.~al.~\cite[Theorem~1.5]{DBLP:conf/stoc/DalirrooyfardMW24}, which finds a single clique of size~4 in time \(\Oh{n^{\omega+1}}\).
	The second algorithm is the one of Eisenbrand and Grandoni~\cite[Proposition~1]{eisenbrand2004fixedparameterclique}, which enumerates all cliques of size~4 in time \(\Oh{m^2}\).
	This allows us to find the set \(U\) in time \(\Oh{\cliqueTimeAddition}\).
	Set \({E^*=E(G[U])}\).
	
	Compute the connected components of \(\etgraph{G}-\trias{U}\) with a breadth-first search in time \(\Oh{|V(\etgraph{G}|)}\).
	Check condition~\labelcref{item:k4-chara-3-in-comp} of \cref{thm:k4-elimination} by counting the elements of \(E^*\) in each connected component.
	If the conditions is satisfied, Maker wins on \(\gamehg{G}\).
	Otherwise, determine the winner of the games on \({\gamehg{G[E(G)\cap K_e]}^{+e}}\) for \(e\in E^*\), where \(K_e\) is the connected component of \(e\) in \(\etgraph{G}-N[E^*\setminus\{e\}]\).
	Each connected component of \({\gamehg{G[E(G)\cap K_e]}^{+e}}\) contains a~2-edge and is thus not~3-uniform.
	Thus, this can be done in time \(\Oh{m+t}\) using \cref{cor:kutz-non-3-uniform}.
	If Maker wins \(\gamehg{G[E(G)\cap K_e]}^{+e}\) for any \(e\in E^*\), Maker wins on \(\gamehg{G}\) by condition~\labelcref{item:k4-chara-biased-win} of \cref{thm:k4-elimination}.
	
	Otherwise, only condition~\labelcref{item:k4-chara-good-anchor} can lead to Maker's win.
	If \(\etgraph{G}-\trias{U}\) does not contain a connected component containing exactly two elements \(e_1,e_2\) of \(E^*\), Breaker wins on \(\gamehg{G}\).
	Otherwise, the winner of the game on \(\gamehg{G}\) is the same as the winner of the Maker-Breaker game played on \(\gamehg{G}+\{e_1,e_2\}\).
	The hypergraph \(\gamehg{G}+\{e_1,e_2\}\) is connected and almost~3-uniform.
	However, it may not be linear due to the new edge \(\{e_1,e_2\}\).
	In this case, remove the single edge that is a strict superset of \(\{e_1,e_2\}\) to obtain an equivalent hypergraph \(H\).
	Otherwise, let \(H\) be \(\gamehg{G}+\{e_1,e_2\}\).
	The outcome of the game on \(H\) can be determined in time \(\Oh{m+t}\) using \cref{cor:kutz-non-3-uniform}.

	The algorithm of \cref{cor:kutz-non-3-uniform} is executed at most~7 times.
	Note that \(z=\Oh{m^{1.5}}\), so the summand \(t\) may be omitted.
	This results in the stated running time.
\end{proof}

\subsection{Characterizing Maker's win when \textsf{ET}(G) is a Cactus Graph}\label{sec:cactus-alg}
Recall that a graph is a cactus graph if all its cycles are edge-disjoint.
The following structures are decisive for Maker's win when \(\etgraph{G}\) is a cactus graph:

\begin{definition}\label{defn:tadpole-fork}
	Let \(G\) be a graph.
	A tadpole \((C,P,e)\) in \(\etgraph{G}\) consists of a cycle \(C\), a path \(P\), possibly consisting only of a single vertex, and an edge-vertex \(e\in E(G)\) such that \({N[\triasCap{C}]\cap N[\triasCap{P}]\subseteq\{e\}}\) and \({V(C)\cap V(P)=\{e\}}\).
	An edge \(e'\in E(G)\) lies on a tadpole if there exists a tadpole \((C,P,e)\) such that \(e'\) is contained in \(V(C)\cup V(P)\).
	A tadpole-fork \((C_1,C_2,P,e_1,e_2)\) in \(\etgraph{G}\) consists of two cycles \(C_1\) and \(C_2\) as well as a path \(P\), possibly consisting only of a single vertex, with the following additional two properties:
	\begin{minipage}{.499\linewidth}
		\begin{itemize}
			\item For \(i\in\{1,2\}\), \((C_i,P,e_i)\) is a tadpole.
		\end{itemize}
	\end{minipage}
	\hfill
	\begin{minipage}{.499\linewidth}
		\begin{itemize}
			\item \(N[\triasCap{C_1}]\cap N[\triasCap{C_2}]\subseteq V(P)\).\qedhere
		\end{itemize}
	\end{minipage}
\end{definition}

The names tadpole and tadpole-fork follow the terminology of~\cite{galliot2025makerbreakersolvedpolynomialtime}.
A tadpole-fork is a fork where the involved dangers are tadpoles, motivating the name.
The second condition of a tadpole-fork ensures that Breaker cannot respond to both threats caused by Maker claiming an edge-vertex of \(P\) at the same time.
An example of a tadpole-fork is displayed in \cref{fig:tadpole-fork}.
\Cref{thm:cactus-game-chara} characterizes Makers win when the edge-triangle incidence graph is a cactus graph using tadpole-forks.

\begin{figure}[!t]
	\centering
	\includegraphics[scale=0.7]{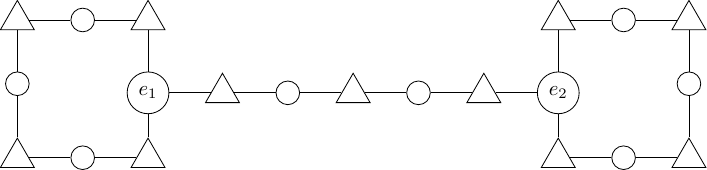}
	\caption{A tadpole-fork.
		Edge-vertices with degree~1 are not drawn.}
	\label{fig:tadpole-fork}
\end{figure}

\begin{theorem}\label{thm:cactus-game-chara}
	Let \(G\) be a graph such that \(\etgraph{G}\) is a cactus graph.
	Then Maker wins on \(\gamehg{G}\) if and only if \(\etgraph{G}\) contains a tadpole-fork.
\end{theorem}
\begin{proof}
	Let \(E_b^i\) be the edges claimed by Breaker up to and including round \(i\).
	We first prove that the presence of a tadpole-fork in \(\etgraph{G}\) is sufficient for Maker's win.
	Let \((C_1,C_2,P,e_1,e_2)\) be a tadpole-fork.
	Maker first claims \(e_1\).
	We distinguish the cases \(e_1=e_2\) and \(e_1\neq e_2\).
	In the first case, the properties of a tadpole-fork imply \(N[\triasCap{C_1}]\cap N[\triasCap{C_2}]=\{e_1\}\), so Breaker cannot claim any edge that will destroy both \(C_1\) and \(C_2\).
	Thus, \(e_1\) lies on a cycle in \(\etgraph{G}-N[E_b^i]\).
	By \cref{lem:maker-winning-structure}, Maker then wins by claiming any edge \(e_1\neq e_3\in E\setminus E_b^1\) that lies on a shortest cycle with \(e_1\) in \(\etgraph{G}-N[E_b^i]\) as Maker will have claimed two edges in the same connected component of \(\etgraph{G}-N[E_b^2]\) regardless of Breaker's second chosen edge.

	In the second case, Breaker must claim an edge in \(N[\triasCap{C_1}]\) as otherwise Maker will achieve victory by claiming an edge that lies on a shortest cycle with \(e_1\).
	As \({N[\triasCap{C_1}]\cap N[\triasCap{C_2}]\subseteq V(P)}\), Breaker cannot have claimed an element of \(N[\triasCap{C_2}]\), so Maker claims \(e_2\).
	If Breaker does not respond with an element of \(N[\triasCap{P}]\), Maker wins by \cref{lem:maker-winning-structure} as \(e_1\) and \(e_2\) both lie on \(P\).
	As \(N[\triasCap{C_2}]\cap N[\triasCap{P}]\subseteq\{e_2\}\), \(e_2\) lies on a cycle in \(\etgraph{G}-N[E_b^2]\).
	By \cref{lem:maker-winning-structure}, Maker wins by claiming an edge that lies on a shortest cycle containing \(e_2\).
	
	We now show that a tadpole-fork is necessary for Maker's win.
	Assume for contradiction that Maker wins on \(\gamehg{G}\) although \(\etgraph{G}\) does not contain a tadpole-fork.
	By \cref{lem:only-cycle-wc,lem:only-cycle-first}, there exists a winning strategy \(M\) for Maker that only claims edges on cycles until (WC) is achieved.
	We inductively construct a Breaker strategy that only allows \(M\) to claim at most a single edge in any connected component of \(\etgraph{G}-E_b^i\) after any round \(i\) while also ensuring that all edges claimed by \(M\) do not lie on a tadpole in \(\etgraph{G}-E_b^i\).
	Then \(M\) never achieves (WC), a contradiction.
	Setting \(E_b^0\coloneqq\emptyset\), this is clearly satisfied in round~0.
	Let the strategy already be constructed up to round \(i-1<m/2\).
	In round \(i\), Maker claims the edge \(e_i\).
	As \(\etgraph{G}\) is a cactus graph without a tadpole-fork, cycles in \(\etgraph{G}\) are vertex-disjoint.
	In particular, there is a unique cycle \(C_1\) containing \(e_i\).
	Let \(P_1\) be the path consisting only of the edge-vertex \(e_i\).
	
	First, assume \(e_i\) is the only edge-vertex claimed by Maker in \(\concomp{e_i}{E_b^{i-1}}\).
	If \(e_i\) only lies on the tadpole \((C_1,P_1,e_i)\), Breaker claims an arbitrary edge-vertex from \(V(C_1)\).
	Then \(e_i\) does not lie on a tadpole in \(\etgraph{G}-N[E_b^i]\).
	Now assume \(e_i\) lies on a tadpole \((C_2,P_2,e_2)\) different from \((C_1,P_1,e_i)\).
	Let \(P_2=v_1\dots v_\ell\).
	Because \(C_1\) is the unique cycle containing \(e_i\), \(e_i\) cannot be in \(V(C_2)\), so \(e_i\in V(P_2)\).
	Thus, we may assume that \(P_2\) begins in \(e_i\), i.e.~\(v_1=e_i\).
	Assume for contradiction that \(v_2\notin V(C_1)\).
	Then \(N[\triasCap{P_2}]\cap N[\triasCap{C_1}]=\{e_i\}\) as otherwise \(e_i\) lies on at least two cycles, which is a impossible.
	This implies \(N[\triasCap{C_1}]\cap N[\triasCap{C_2}]=\emptyset\) as otherwise \(e_i\) would again be contained in at least two cycles.
	But then \((C_1,C_2,P_2,e_i,e_2)\) is a tadpole-fork, a contradiction.
	Thus, \(v_2\in V(C_1)\).
	As \(P_2\) ends in an edge-vertex and \(v_2\) is a triangle-vertex, we obtain \(\ell\geq3\).
	Breaker claims \(v_3\).
	Assume for contradiction that \(e_i\) lies on a tadpole \((C_3,P_3,e_3)\) in \(\etgraph{G}-N[E_b^i]\).
	As before, we may assume that \(P_3\) begins in \(e_i\).
	Set \(Q=P_3v_2\dots v_\ell\).
	Then \((C_2,C_3,Q,e_2,e_3)\) is a tadpole-fork in \(\etgraph{G}\), a contradiction.

	Now assume \(\concomp{e_i}{E_b^{i-1}}\) contains another Maker edge \(e^*\).
	Let \(P=v_1\dots v_\ell\) be a shortest path from \(V(C_1)\) to \(e^*\).
	By construction of the strategy, \(e^*\) does not lie on a tadpole, so \(v_1\in\trias{G}\).
	Breaker claims an element of \(N(v_1)\setminus\{e_i,e^*\}\).
	This separates \(e_i\) and \(e^*\) in \(\etgraph{G}-N[E_b^i]\).
	Assume for contradiction that \(e_i\) lies on a tadpole \((C_2,P_2,e_2)\) in \(\etgraph{G}-N[E_b^i]\).
	Let \(P_3\) be a shortest path from \(e^*\) to \(e_2\) in \(\etgraph{G}-N[E_b^{i-1}]\).
	As \(C_2\) is the unique cycle containing \(e_2\), \(N[\triasCap{P_3}]\) can only intersect \(V(C_2)\) in a single vertex, namely \(e_2\).
	Thus, \((C_2,P_3,e_2)\) is a tadpole in \(\etgraph{G}-N[E_b^{i-1}]\), contradicting the construction of the strategy.
\end{proof}

Using \cref{thm:cactus-game-chara}, it is possible to determine the winner of the game on \(\gamehg{G}\) without relying on \cref{thm:kutz-breaker-chara} by checking the presence of tadpole-forks in \(\etgraph{G}\).
This is achieved by iterating over the 2-connected components of each connected component in reverse order of discovery by a breadth-first search.
In a cactus graph \(G\), every 2-connected component of size at least two is a cycle.
Assume \(G\) is connected, but not a tree, and let \(H\) be the following graph:
The vertices of \(H\) are the cycles of \(G\) and there exists an edge between two cycles if and only if there exists a path between the two cycles that does not intersect any other cycle in \(G\).
Then \(H\) is a tree and we call it the cycle-cut-tree of \(G\).

\begin{theorem}\label{thm:cactus-time}
	Let \(G\) be a graph such that \(\etgraph{G}\) is a cactus graph.
	Then the winner of the game on \(\gamehg{G}\) can be determined in time \(\Oh{\cactusTime}\).
\end{theorem}
\begin{proof}
	Compute \(\etgraph{G}\) using \cref{lem:computing-etg}.
	Iterate over \(E(\etgraph{G})\) to determine the degree of every vertex.
	Compute the connected components of \(\etgraph{G}\) using breadth-first search and execute the following for each connected component \(K\):
	
	\begin{enumerate}[label=\arabic*]
		\item Compute the maximal 2-connected components \(B_1,\dots,B_k\) of \(K\) with size at least~3.
		These components are the cycles of \(K\).
		\item If \(k\leq 1\), move on to the next connected component or return that \(G\) is Breaker's win if \(K\) is the last remaining component.
		\item\label{alg:cactus_deg_pruning} Execute a BFS in \(K\) starting from an arbitrary vertex of degree at least~2.
		Iterate over the vertices in reverse BFS order, removing any vertex of degree~1 and updating the degree of the parent.
		\item Iterate over each cycle counting the number of vertices of degree at least~3.
		\item Execute a BFS starting in a cycle with the maximal number of vertices of degree at least~3.
		\item\label{alg:cactus_cuttree_order} Iterate over the cycles in reverse order of discovery by the BFS:
		\begin{enumerate}[label=\arabic*]\setcounter{enumii}{\value{enumi}}
			\item Let \(B_i\) be the current cycle.
			If \(B_i\) does not contain a vertex of degree~3, move on to the next connected component or return that \(G\) is Breaker's win if \(K\) is the last remaining component.
			\item\label{alg:cactus_ancestor_removal} If the unique vertex \(v\) of degree~3 in \(B_i\) is a triangle-vertex, remove \(B_i\) from \(K\) and follow the tree edges of the BFS while removing the ancestors of \(v\) until the first ancestor on a cycle is found.
			Then move on to the next cycle.
			\item If \(v\) is an edge-vertex, execute a BFS from \(v\).
			If the first vertex of any cycle discovered by the BFS is an edge-vertex, \(G\) is Maker's win.
			Otherwise, move on to the next connected component of \(\etgraph{G}\) or return that \(G\) is Breaker's win if \(K\) is the last remaining component.
		\end{enumerate}
	\end{enumerate}
	
	The algorithm executes a constant number of breadth-first searches in each connected component of \(\etgraph{G}\).
	The 2-connected components can be computed in linear time~\cite{DBLP:journals/cacm/HopcroftT73}.
	Note that \({t=\Oh{m^{1.5}}}\).
	Thus, we obtain the stated running time bound.
	
	We now show correctness of this algorithm.
	Step~\labelcref{alg:cactus_deg_pruning} ensures that all remaining vertices of \(K\) have degree at least~2 by pruning low degree vertices.
	Because of this, every vertex lies on a cycle or connects two different cycles.
	In particular, the vertices of degree at least~3 connect different cycles.
	
	The iteration order in step~\labelcref{alg:cactus_cuttree_order} corresponds to the reverse BFS order when starting a BFS in a vertex with maximum degree in the cycle-cut-tree of \(\etgraph{G}\).
	This order ensures that each cycle contains at most one vertex of degree at least~3 and this vertex has degree exactly~3 at the time the cycle is processed:
	This is clear for the first processed cycle as otherwise the breadth-first search would not have ended in that cycle.
	The removal of the ancestors of a vertex in step~\labelcref{alg:cactus_ancestor_removal} then ensures that the parent cycle of \(B_i\) will contain at most one vertex of degree at least~3 by the time it is processed.
	The removal step can be seen as pruning a leaf in the cycle-cut-tree of \(\etgraph{G}\).
	
	Any cycle removed by the algorithm cannot contribute to a tadpole-fork as it is connected to the rest of the graph by a triangle-vertex.
	If a tadpole-fork exists in a connected component \(K\), then every cycle connected to the rest of the component by an edge-vertex also lies on a tadpole-fork.
	By \cref{thm:cactus-game-chara}, the presence of a tadpole-fork in \(\etgraph{G}\) determines the winner of the game, so the algorithm is correct.
\end{proof}

\subsection{Lower Bound by Triangle Detection}
\begin{figure}[t]
	\centering
	\begin{equation*}
		\begin{array}{c@{\hspace{10mm}}c}
			\vcenter{\hbox{\includegraphics[scale=0.7]{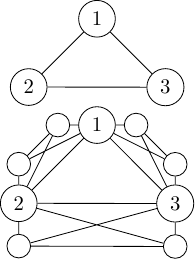}}} &
			\vcenter{\hbox{\includegraphics[scale=0.6]{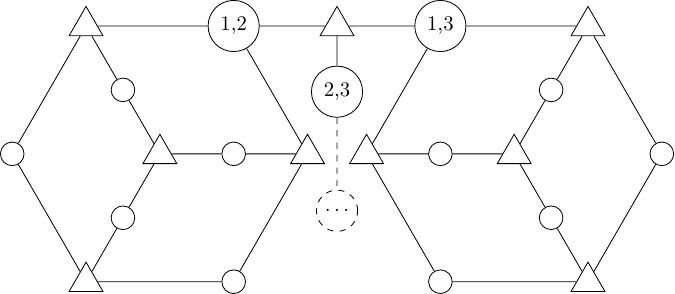}}} \\[2pt]
			G_1 \text{ and } H_1 & \etgraph{H_1} \\[4pt]
			\vcenter{\hbox{\includegraphics[scale=0.7]{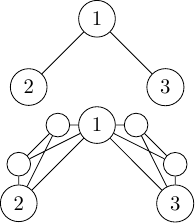}}} &
			\vcenter{\hbox{\includegraphics[scale=0.6]{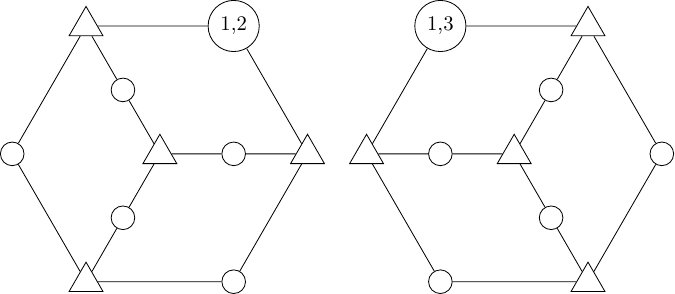}}} \\[2pt]
			G_2 \text{ and } H_2 & \etgraph{H_2}
		\end{array}
	\end{equation*}
	\caption{The graphs \(G_1\) and \(G_2\), their respective image \(H_1\) and \(H_2\) under the reduction from \cref{thm:trig-detection-lb} as well as \(\etgraph{H_1}\) and \(\etgraph{H_2}\).
		The edge-triangle incidence graph \(\etgraph{H_1}\) is only partially drawn.
		The dashed vertex represents the remainder of the copy of \(\etgraph{K_4}\) that uses the edge \(\{2,3\}\).}
	\label{fig:trig_reduction}
\end{figure}

\begin{theorem}\label{thm:trig-detection-lb}
	There exists an algorithm that, given any graph \(G\), computes a graph \(H\) in linear time such that \(K_3\subseteq G\) if and only if Maker wins on \(\gamehg{H}\).
	The graph \(H\) has \(n+2m\) vertices and \(6m\) edges.
\end{theorem}
\begin{proof}
	The graph \(H\) is obtained from \(G\) by adding two new vertices \(x_e\) and \(y_e\) for each edge \(vw=e\in E(G)\) as well as the edges \(x_ey_e,x_ev,x_ew,y_ev\) and \(y_ew\), i.e.~each edge of \(G\) is augmented to a copy of \(K_4\) through two new vertices.
	The running time of this process depends only linearly on the number of edges in \(G\).
	
	If \(G\) contains a triangle \(t\), then \(t\) connects two disjoint copies of \(\etgraph{K_4}\) in \(\etgraph{H}\).
	Thus, \(H\) is winning for Maker by condition~\labelcref{item:k4-chara-biased-win} of \cref{thm:k4-elimination}.
	If \(G\) does not contain a triangle, then \(\etgraph{H}\) consists of disjoint copies of \(\etgraph{K_4}\).
	By \cref{cor:etg-connected,thm:k4-elimination}, the graph \(H\) is Breaker's win.
	See \cref{fig:trig_reduction} for an illustration of these two cases.
\end{proof}

\Cref{thm:trig-detection-lb} shows that deciding the outcome of the unbiased triangle game is at least as hard as triangle detection.
Thus, it is unlikely that any algorithm can decide the outcome of the unbiased triangle game in time \(\oh{n^\omega}\).

\section{Edge Bounds for Maker Wins}\label{sec:min-maker-win}
\begin{definition}
	A graph \(G\) is minimal for Maker if it is Maker's win and all proper subgraphs of \(G\) are Breaker's win.
	We will also say that \(G\) is a minimal Maker win.
\end{definition}

In this section, we give an asymptotically exact lower bound on the number of edges of a graph required to ensure that Maker wins the unbiased triangle game, see \cref{thm:asymptotic-density}.
Additionally, we give a tight lower bound on the number of edges in a minimal Maker win, see \cref{thm:min-maker-edge-bound}.

\begin{theorem}\label{thm:asymptotic-density}
	Let \(\varepsilon>0\).
	There exists a constant \(n_0\in\N\) such that for all graphs \(G\) satisfying \({n\geq n_0}\) and \({m\geq\left(\frac{1}{4}+\varepsilon\right)n^2}\) Maker wins on \(\gamehg{G}\).
	This is best possible in the sense that there exists an infinite family of graphs with exactly \(\frac{n^2}{4}\) edges such that Breaker wins the unbiased \(K_3\)-game on this family.
\end{theorem}
\begin{figure}
	\centering
	\includegraphics[scale=0.7]{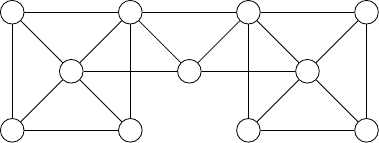}
	\caption{A 3-colorable graph that is Maker's win.}
	\label{fig:density-win}
\end{figure}
\begin{proof}
	Consider the graph \(G\) depicted in \cref{fig:density-win}.
	Its edge-triangle incidence graph \(\etgraph{G}\) is depicted in \cref{fig:tadpole-fork}.
	Observe that \(\etgraph{G}\) is a cactus graph containing a tadpole-fork.
	By \cref{thm:cactus-game-chara}, \(G\) is Maker's win.
	As \(G\) is 3-colorable, the Erdős-Stone-Theorem~\cite{DBLP:books/daglib/0030488} yields that every graph \(H\) with at least
	\[\left(1-\frac{\chi(G)-2}{\chi(G)-1}+\oh{1}\right)\frac{n^2}{2}=\left(\frac{1}{4}+\oh{1}\right)n^2\]
	edges contains a copy of \(G\).
	The vanishing function of \(n\) suppressed by \(\oh{1}\) depends only on \(G\) and not on \(H\).
	Breaker wins the unbiased triangle game on any bipartite graph as there are no triangles in such graphs.
	The infinite family \(F\) is given by the balanced complete bipartite graphs, i.e.~\({F=\{K_{n,n}:n\in\N\}}\).
	Such graphs satisfy \(|E(K_{n,n})|=n^2=|V(K_{n,n})|^2/4\).
\end{proof}

The lower bound \(m\geq 2n-1\) for minimal Maker wins does not require asymptotic arguments.
There exists a graph with \(n\) vertices satisfying the bound with equality for every \(n\geq5\), so the bound is tight.
This is best possible as all graphs with at most four vertices are winning for Breaker.
The inequality is not always satisfied with equality as witnessed by the graph displayed in \cref{fig:min-win-not-strict}.

\begin{figure}[t]
	\renewcommand\tabularxcolumn[1]{m{#1}}
	\centering
	\begin{tabularx}{\linewidth}{*{5}{>{\centering\arraybackslash}X}}
		\includegraphics[scale=0.5]{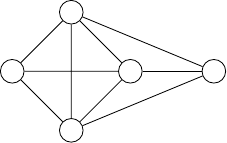}
		&
		\includegraphics[scale=0.5]{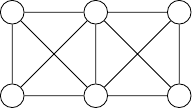}
		&
		\includegraphics[scale=0.5]{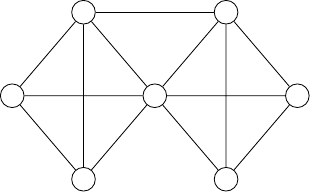}
		&
		\includegraphics[scale=0.5]{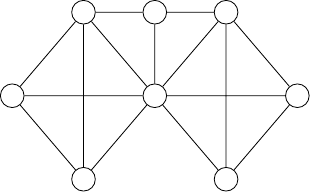}
		&
		\includegraphics[scale=0.5]{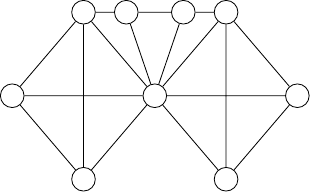}\\
		$H_5$ & $H_6$ & $H_7$ & $H_8$ & $H_9$
	\end{tabularx}
	\caption{The graphs \(H_n\) for \(5\leq n\leq 9\).
	They are minimal Maker wins satisfying \(m=2n-1\).
	For \(n\geq 7\), the graph \(H_{n+1}\) is obtained from \(H_n\) in the same way \(H_8\) is obtained from \(H_7\).}
	\label{fig:min-maker-equality}
\end{figure}

\begin{figure}[!t]
	\centering
	\includegraphics[scale=0.7]{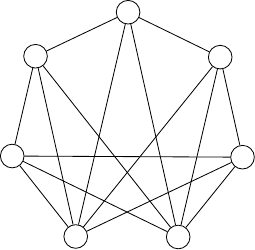}
	\caption{A minimal Maker win with 7 vertices and 15 edges.}
	\label{fig:min-win-not-strict}
\end{figure}

\begin{theorem}\label{thm:min-maker-edge-bound}
	Let \(G\) be a minimal Maker win.
	Then \(m\geq 2n-1\).
	For every \(n\geq 5\), there exists a minimal Maker win \(H_n\) with \(n\) vertices and \(2n-1\) edges.
\end{theorem}
\begin{proof}
	Our proof of the first statement is based on the ideas of the proof of~\cite[Theorem~1.3]{DBLP:journals/rsa/MullerS14}.
	However, as we are not in a probabilistic setting, we cannot dismiss cases that will not happen asymptotically almost surely.
	By \cref{cor:etg-connected}, there exists a connected component \(K\) of \(\etgraph{G}\) such that \(G=G[E(G)\cap K]\) as otherwise edges of \(G\) could be removed, contradicting its minimality.
	In particular, every vertex and every edge of \(G\) lies on a triangle.
	Thus, \(\delta(G)\geq2\).
	Additionally, if a vertex \(v\) has degree exactly two, \cref{cor:lowdeg-pruning} implies that the edges \(e_1,e_2\) incident to \(v\) can be removed from the vertex set of \(\gamehg{G}\) without changing the winner of the Maker-Breaker game played on \(\gamehg{G}\).
	This is equivalent to removing these edges in \(G\), so Maker also wins on \(\gamehg{G-\{e_1,e_2\}}\), again contradicting minimality of \(G\).
	We conclude that \(\delta(G)\geq3\).
	
	Cover the vertices of \(G\) with triangles in the following way:
	First, choose an arbitrary triangle \({t_1\in\trias{G}}\).
	Let \(t_1,\dots,t_i\) be chosen already.
	While this is still possible, pick a new triangle $t_{i+1}$ such that \({|t_{i+1}\cap\bigcup_{j=1}^it_j|=1}\).
	Let \(t_1,\dots,t_\ell\) be the triangles at the end of this process, i.e~\(\ell\) is maximal.
	Assume for contradiction that there exists some \(v\in V(G)\setminus\{v:v\in e\in t_i,1\leq i\leq\ell\}\).
	Let \(t^*\) be a triangle containing \(v\) and let \(P=w_1\dots w_k\) be a shortest path from \(\{t_1,\dots,t_\ell\}\) to \(t^*\) in \(\etgraph{G}\).
	All vertices of \(P\) with odd index are triangles and two triangles \(w_{2i+1},w_{2i+3}\) intersect in the edge \(w_{2i+2}\).
	Thus, they share at least one edge of \(G\) and they cannot share two edges as otherwise they would be identical.
	Let \(j\) be minimal such that \(w_{2j+1}\setminus\bigcup_{j=1}^\ell t_j\neq\emptyset\).
	Then \(w_{2j-1}\subseteq\bigcup_{j=1}^\ell t_j\) and \(|w_{2j-1}\cap w_{2j+1}|=1\).
	Thus, \(w_{2j+1}\) can be used to expand the sequence \(t_1,\dots,t_\ell\), which is a contradiction.
	
	Let \(G'\) be the subgraph of \(G\) induced by the above covering process.
	Formally, \(V(G')=V(G)\) and \(E(G')=\bigcup_{j=1}^\ell t_j\).
	As each triangle used to construct \(G'\) covers a new vertex of \(V(G)\), we obtain \({n=|V(G')|=\ell+2}\) and \(|E(G')|=2\ell+1\).
	Set \(d=m-|E(G')|\).
	By construction of \(G'\), \(\etgraph{G'}\) is a tree.
	Thus, \(\etgraph{G'}\) is a cactus graph without a tadpole-fork, so Breaker wins on \(\gamehg{G}\) by \cref{thm:cactus-game-chara}.
	However, \(G\) is Maker's win, so \(G'\) must be a proper subgraph of \(G\) which implies \(d\geq 1\).
	
	Assume \(d=1\).
	A simple inductive argument shows that \(G'\) contains at least two vertices of degree~2.
	As \(\delta(G)\geq3\), there must be exactly two vertices of degree~2 in \(G'\) which are the endpoints of an edge in \(E(G)\setminus E(G')\).
	The only way this is possible is when \(G\) is a wheel.
	Each connected component of \(\etgraph{K_4}-\trias{K_4}\) is an isolated edge-vertex, so Breaker wins on \(\gamehg{K_4}\) by~\cref{thm:k4-elimination}.
	Thus, \(G\) cannot be the wheel \(K_4\).
	Then \(\etgraph{G}\) contains only a single cycle, so it is cactus graph that does not contain a tadpole-fork.
	By \cref{thm:cactus-game-chara}, Breaker wins on \(\gamehg{G}\), a contradiction.
	This yields \(d\geq 2\).
	We obtain
	\begin{equation*}
		m=|E(G')|+d\geq2\ell+3=2\ell+4-1=2n-1,
	\end{equation*}
	which concludes the proof of the first part of \cref{thm:min-maker-edge-bound}.
	
	For \(5\leq n\leq 9\), the graphs \(H_n\) are visualized in~\cref{fig:min-maker-equality}.
	The graph \(H_5\) is given by \(K_5\) with a single edge removed.
	The graph \(H_6\) is given by two copies of \(K_4\) intersecting in a single edge.
	For \(n\geq7\), the graph \(H_n\) is obtained by connecting two copies of \(K_4\) intersecting in exactly one vertex with a sequence of \(n-6\) triangles involving \(n-7\) additional vertices such that \(\etgraph{H_n}\) contains a tadpole-fork.
	The graph \(H_n\) can easily be seen to be a minimal Maker win using \cref{thm:k4-elimination,cor:lowdeg-pruning}.
\end{proof}

\section{Open Problems}\label{sec:opem-problems}
The lower bound by triangle detection does not match any of the upper bounds.
Closing this gap by increasing the lower bound through a different fine-grained reduction or reducing the upper bound through novel algorithmic ideas is an interesting direction for future work.

\Cref{thm:asymptotic-density} determines an asymptotic lower bound on the average degree of the vertices of the board graph required to ensure Maker's win.
The family witnessing that \cref{thm:asymptotic-density} is best possible consists of bipartite graphs.
This seems somewhat artificial as an obvious preprocessing step in any algorithm for deciding the unbiased triangle game is to remove all vertices and edges that do not lie on a triangle.
In fact, \cref{cor:etg-connected} shows that it is enough to individually inspect the subgraphs of the board graph induced by the edges in the connected components of the edge-triangle incidence graph.
For bipartite graphs, these subgraphs consist of isolated edges.
Additionally, \cref{cor:lowdeg-pruning} can be applied to further eliminate edges.
Eliminating edges decreases the average degree and taking a subgraph is not guaranteed to increase the average degree.
However, applying \cref{cor:lowdeg-pruning} does increase the triangle-edge-ratio, i.e.~\(|\trias{G}|/|E(G)|\), given that the board graph contains at least two triangles.
Similarly, there must be a connected component of the edge-triangle incidence graph that induces a graph which has at least the same triangle-edge-ratio as the board graph itself.
Thus, we arrive at the following problem:
Given \(5\leq n_0\in\N\), determine a function \(f:\N\to\N\) such that any graph \(G\) with \({n\geq n_0}\) and \(t\geq f(n)\) such that \(\etgraph{G}\) is connected and cannot be reduced by \cref{cor:lowdeg-pruning} is Maker's win and show that no function \(g\) with \(g(n)<f(n)\) for some \(n_0\leq n\in\N\) satisfies the same property.
It is clear that such a function \(f\) exists, but it may be too difficult to determine it exactly.
An alternative problem is to determine the asymptotic behavior of the function \(f\).

A similar problem can be asked for the average degree.
Imposing the additional restriction that \(\etgraph{G}\) be connected and irreducible by \cref{cor:lowdeg-pruning}, what is the asymptotic average degree required to ensure Maker's win?
Clearly, it is at most the same value as determined by \cref{thm:asymptotic-density}, but the family witnessing the lower bound is no longer applicable in this situation.

Finally, there is also the unbiased \(C_4\)-game or more generally the unbiased \(C_n\)-game for \(n\geq4\).
No polynomial-time algorithm is known for these games, but it is also not known if the games are complete for any complexity class.
It was shown in~\cite{DBLP:journals/dam/DucheneGINOPS25,DBLP:journals/corr/abs-2509-13819} that the unbiased \(H\)-game is \pspace-complete for a graph \(H\) with~35 vertices and~39 edges.
Additionally, the \(C_4\)-game played on the \emph{vertices} instead of the edges of general graphs was shown to be \pspace-complete in \cite{DBLP:journals/corr/abs-2509-13819}.
Therefore, it seems likely that the unbiased \(C_n\)-game played on edges is also \pspace-complete for \(n\geq4\).

\section*{Acknowledgments}
Julian Brinkmann is the main author of this paper and contributed a more than proportional amount of the material.
The first author would like to thank the authors of the JGraphT library~\cite{DBLP:journals/toms/MichailKNS20} which was useful in the exploratory phase of research.

\clearpage
\bibliographystyle{plainurl}
\bibliography{mb}
\end{document}